\documentclass[reqno,12pt]{amsart}
\usepackage{amsmath,amssymb,ifthen,bbm,subfigure,verbatim,ifthen}
\usepackage{url}
\usepackage{color}
\usepackage{latexsym}
\usepackage{enumerate}
\usepackage[latin1]{inputenc}
\usepackage[round,authoryear]{natbib}
\usepackage{ifpdf}
\usepackage{graphicx}
\ifpdf 
\usepackage[pdftex, 
bookmarks = true, 
bookmarksnumbered = true, 
hyperindex = true,			
linktocpage = true,			
pdfpagemode = UseNone, 
pdfstartview = FitH, 
pdfpagelayout = SinglePage,
colorlinks = true, 
urlcolor = blue, 
citecolor = blue,
]{hyperref}
\else 
\usepackage[ps2pdf, 
breaklinks,
bookmarks = true, 
bookmarksnumbered = true, 
hyperindex = true,			
linktocpage = true,			
pdfpagemode = UseNone, 
pdfstartview = FitH, 
pdfpagelayout = SinglePage,
colorlinks = true, 
urlcolor = blue, 
citecolor = blue,
]{hyperref}
\fi
\usepackage[section]{algorithm}
\usepackage{algorithmicx}
\usepackage{algpseudocode}

\def\href#1#2{#2}

\newcommand{\Note}[2]{}
\newcommand{\CSD}{\mathrm{CSD}}
\newcommand{\KLD}{\mathrm{KLD}}
\newcommand{\ie}{i.e.}
\newcommand{\eg}{e.g.}
\newcommand{\wrt}{with respect to}
\newcommand{\rset}{\mathbb{R}}
\newcommand{\adj}[1]{\psi_{N,#1}}
\newcommand{\adjopt}[1]{\psi^\ast_{N,#1}}

\newcommand{\adjfunc}{\Psi}
\newcommand{\adjfuncoptKL}{\Psi^\ast_{\mathrm{KL},\prop}}
\newcommand{\adjfuncoptchi}{\Psi^\ast_{\chi^2,\prop}}
\newcommand{\alg}{\mathcal{B}}
\newcommand{\argmax}{\operatornamewithlimits{arg\,max}}
\newcommand{\argmin}{\operatornamewithlimits{arg\,min}}
\newcommand{\auxinstr}{\pi_{N}^{\mathrm{aux}}}
\newcommand{\auxinstrparam}[1]{\pi^{\mathrm{aux}}_{N,#1}}
\newcommand{\auxtarg}{\mu_N^{\mathrm{aux}}}

\newcommand{\chitwo}{d_{\chi^2}}

\newcommand{\CV}{\ensuremath{\mathrm{CV}^2}}
\newcommand{\define}{\triangleq}

\newcommand{\distrn}{\lambda}
\newcommand{\E}{\mathbb{E}}
\newcommand{\entropy}{\ensuremath{\mathcal{E}}}

\newcommand{\eqsp}{\;}

\newcommand{\hk}{Q}
\newcommand{\ind}[2][]{I_{N,#2}\ifthenelse{\equal{#1}{}}{}{^{#1}}}

\newcommand{\indic}{\mathbbm{1}}
\newcommand{\instr}{\pi_N}

\newcommand{\KL}{d_{\mathrm{KL}}}
\newcommand{\Lp}[1]{\mathsf{L}^{#1}}
\newcommand{\meanopt}{\tau}

\newcommand{\meas}{\mathcal{P}}
\newcommand{\measfunc}{\mathbb{B}}

\newcommand{\noiseletter}{\epsilon}
\newcommand{\noise}[1]{\noiseletter_{#1}}
\newcommand{\normdens}[3]{\mathcal{N}(#1; #2, #3)}
\newcommand\param{%
	\bgroup\theta \paraminterne}
\newcommand\paraminterne[1][]{%
\ifthenelse{\equal{#1}{}}{}{_N^{#1}}\egroup}
\newcommand{\paramin}{\theta_N^\ast}
\newcommand{\paramsp}{\Theta}
\newcommand{\parti}[2][]{\xi_{#2} \ifthenelse{\equal{#1}{}}{}{^{#1}}}
\newcommand{\partitd}[2][]{\tilde{\xi}_{N,#2}\ifthenelse{\equal{#1}{}}{}{^{#1}}}

\newcommand{\prob}{\operatorname{\mathbb{P}}}
\newcommand\prop{%
\bgroup R\propinterne}
\newcommand\propinterne[1][]{%
\ifthenelse{\equal{#1}{}}{}{_{#1}}\egroup}
\newcommand{\propdens}[1][]{r\ifthenelse{\equal{#1}{}}{}{_{#1}}}
\newcommand{\propjoint}[1][]{\pi^\ast\ifthenelse{\equal{#1}{}}{}{_{#1}}}
\newcommand{\R}{\mathbb{R}}
\newcommand{\stdopt}{\eta}
\newcommand{\stsp}{{\boldsymbol{\Xi}}}
\newcommand{\stsptd}{{\tilde{\boldsymbol{\Xi}}}}
\newcommand{\targ}{\mu_N}
\newcommand{\targjoint}{\mu^\ast}
\newcommand{\ud}{\mathrm{d}}
\newcommand{\uk}{L}

\newcommand{\ukdens}{l}

\newcommand{\wgt}[1]{\omega_{#1}}

\newcommand{\wgtfunc}{\Phi}
\newcommand{\wgtfuncIS}[1][]{W \ifthenelse{\equal{#1}{}}{}{_{#1}}}
\newcommand{\compwgtfuncIS}[1][]{w \ifthenelse{\equal{#1}{}}{}{_{#1}}}
\newcommand{\wgttd}[2][]{\tilde{\omega}_{N,#2}\ifthenelse{\equal{#1}{}}{}{^{#1}}}
\newcommand{\wgtsum}{\Omega_N}
\newcommand{\wgtsumtd}[1][]{\tilde{\Omega}_N\ifthenelse{\equal{#1}{}}{}{^{#1}}}
\newcommand{\strans}{q}
\newcommand{\olik}{g}
\newcommand{\prior}{\pi_0}
\newcommand{\post}[2]{\phi_{#1|#2}}
\newcommand{\lhood}[1]{\ensuremath{\ell_{#1}}}

\newcounter{hyp}
\newenvironment{hyp}[1]{\refstepcounter{hyp}\it\begin{itemize}\item[{\bf
      (A\arabic{hyp})}] \label{hyp:#1}}{\end{itemize}}
\newcommand{\refhyp}[1]{{\bf (A\ref{hyp:#1})}}

\newtheorem{definition}{Definition}[section]
\newtheorem{lemma}{Lemma}[section]
\newtheorem{proposition}{Proposition}[section]
\newtheorem{remark}{Remark}[section]
\newtheorem{theorem}{Theorem}[section]
\numberwithin{equation}{section}

\title[Adaptive methods for sequential importance sampling]{Adaptive methods for sequential importance sampling with application to state space models
}

\author[J.~Cornebise]{Julien Cornebise}
\address[J.~Cornebise]{Institut des T\'el\'ecoms, T\'el\'ecom ParisTech \\
              46 Rue Barrault, 75634 Paris Cedex 13, France} 
\email{\href{mailto:julien.cornebise@telecom-paristech.fr}{julien.cornebise@telecom-paristech.fr}} \thanks{This work was partly supported by the
National Research Agency (ANR) under the program ``ANR-05-BLAN-0299''}

\author[\'E.~Moulines]{\'Eric Moulines}
\address[\'E.~Moulines]{Institut des T\'el\'ecoms, T\'el\'ecom ParisTech \\
              46 Rue Barrault, 75634 Paris Cedex 13, France} 
\email{\href{mailto:moulines@telecom-paristech.fr}{moulines@telecom-paristech.fr}}

\author[J.~Olsson]{Jimmy Olsson}
\address[J.~Olsson]{Center of Mathematical Sciences, Lund University \\  Box                  118, SE-22100 Lund, Sweden} 
\email{\href{mailto:jimmy@maths.lth.se}{jimmy@maths.lth.se}}
\date{July 2008}

\begin{document}
\begin{abstract}
In this paper we discuss new adaptive proposal strategies for sequential Monte Carlo algorithms---also known as particle filters---relying on criteria evaluating the quality of the proposed particles. The choice of the proposal distribution is a major concern and can dramatically influence the quality of the estimates. Thus, we show how the long-used coefficient of variation (suggested by \cite{kong:liu:wong:1994}) of the weights can be used for estimating the chi-square distance between the target and instrumental distributions of the auxiliary particle filter. As a by-product of this analysis we obtain an auxiliary adjustment multiplier weight type for which this chi-square distance is
minimal. Moreover, we establish an empirical estimate of linear complexity of the Kullback-Leibler divergence between the involved distributions. Guided by these results, we discuss adaptive designing of the particle filter proposal distribution and illustrate the methods on a numerical example.  \\
\end{abstract}



\maketitle
\section{Introduction}

Easing the role of the user by tuning automatically the key
parameters of \emph{sequential Monte Carlo} (SMC) \emph{algorithms}
has been a long-standing topic in the community, notably through
adaptation of the particle sample size or the way the particles are sampled and weighted.
In this paper we focus on the latter issue and develop methods for adjusting adaptively the proposal distribution of the particle filter.

Adaptation of the number of particles has been trea\-ted by several authors.
In \protect{\cite{legland:oudjane:2006}} (and later \citet[Section
IV]{hu:schon:ljung:2008}) the size of the particle sample is increased until
the total weight mass reaches a positive threshold, avoiding a situation where all particles are located in regions of the state space having zero posterior probability. \citet[Section 3.2]{fearnhead:liu:2007} adjust the size of the particle cloud in order to control the error introduced by the resampling step.
Another approach, suggested by \cite{fox:2003} and refined in \citet{soto:2005} and \cite{straka:simandl:2006}, consists in increasing the sample size until the \emph{Kullback-Leibler divergence} (KLD) between the true and estimated target distributions is below a given threshold.

Unarguably, setting an appropriate sample size is a key ingredient of
any statistical estimation procedure, and there are cases where the
methods mentioned above may be used for designing satisfactorily this size; however increasing the sample size only is far from being always sufficient for achieving efficient variance reduction.
Indeed, as in any algorithm based on importance sampling, a
significant discrepancy between the proposal and target distributions
may require an unreasonably large number of samples for decreasing the
variance of the estimate under a specified value. For a very simple
illustration, consider importance sampling estimation of the mean $m$
of a normal distribution using as importance distribution another normal distribution having zero mean and same variance: in this case, the variance of the estimate grows like $\exp(m^2)/N$, $N$ denoting the
number of draws, implying that the sample size required for ensuring a
given variance grows exponentially fast with $m$.

This points to the need for adapting the importance distribution of the particle filter, e.g., by adjusting at each iteration the particle weights and the proposal distributions; see \eg\  \cite{doucet:godsill:andrieu:2000}, \cite{liu:2004}, and \cite{fearnhead:2008} for reviews of various filtering methods. These two quantities are critically important, since the performance of the particle filter is closely related to the ability of proposing particles in state space regions where the posterior is significant. It is well known that sampling using as proposal distribution the mixture composed by the current particle importance weights and the prior kernel (yielding the classical bootstrap particle filter of \cite{gordon:salmond:smith:1993}) is usually inefficient when the likelihood is highly peaked or located in the tail of the prior.

In the sequential context, the successive distributions to be approximated (\eg\ the successive filtering distributions) are the iterates of a nonlinear random mapping, defined on the space of probability measures; this nonlinear mapping may in general be decomposed into two steps: a prediction step which is linear and a nonlinear correction step which amounts to compute a normalisation factor. In this setting, an appealing way to update the current particle approximation consists in sampling new particles from the distribution obtained by propagating the current particle approximation through this mapping; see \eg\ \cite{huerzeler:kuensch:1998}, \cite{doucet:godsill:andrieu:2000}, and \cite{kuensch:2005} (and the references therein). This sampling distribution guarantees that the conditional variance of the importance weights is equal to zero. As we shall see below, this proposal distribution enjoys other optimality conditions, and is in the sequel referred to as the \emph{optimal sampling distribution}.
However, sampling from the optimal sampling distribution is, except for some specific models, a difficult and time-consuming task (the in general costly auxiliary accept-reject developed and analysed by \cite{kuensch:2005} being most often the only available option).

To circumvent this difficulty, several sub-optimal schemes have been proposed. A first type of approaches tries to mimic the behavior of the optimal sampling without suffering the sometimes prohibitive cost of rejection sampling. This typically involves localisation of the modes of the unnormalised optimal sampling distribution by means of some optimisation algorithm, and the fitting of over-dispersed student's $t$-distributions around these modes; see for example \cite{shephard:pitt:1997}, \cite{doucet:defreitas:gordon:2001}, and \cite{liu:2004} (and the references therein). Except in specific cases, locating the modes involves solving an optimisation problem for every particle, which is quite time-consuming.

A second class of approaches consists in using some classical approximate non-linear filtering tools such as the \emph{extended Kalman filter} (EKF) or the \emph{unscented transform Kalman filter} (UT/UKF); see for example \cite{doucet:defreitas:gordon:2001} and the references therein.  These techniques assume implicitly that the conditional distribution of the next state given the current state and the observation has a single mode. In the EKF version of the particle filter, the linearisation of the state and observation equations is carried out for each individual particle. Instead of linearising the state and observation dynamics using Jacobian matrices, the UT/UKF particle filter uses a deterministic sampling strategy to capture the mean and covariance with a small set of carefully selected points (\emph{sigma points}), which is also computed for each particle. Since these computations are most often rather involved, a significant computational overhead is introduced.

A third class of techniques is the so-called \emph{auxiliary particle filter} (APF)
suggested by \citet{pitt:shephard:1999}, who proposed it as a way to build data-driven proposal distributions (with the initial aim of robustifying standard SMC methods to the presence of outlying observations); see \eg\ \cite{fearnhead:2008}. The procedure comprises two stages: in the first-stage, the current particle weights are modified in order to select preferentially those particles being most likely proposed in regions where the posterior is significant. Usually this amounts to multiply the weights with so-called \emph{adjustment multiplier weights}, which may depend on the next observation as well as the current position of the particle and (possibly) the proposal transition kernels. Most often, this adjustment weight is chosen to estimate the predictive likelihood of the next observation given the current particle position, but this choice is not necessarily optimal.

In a second stage, a new particle sample from the target distribution is formed using this proposal distribution and associating the proposed particles with weights proportional to the inverse of the adjustment multiplier weight \footnote{The original APF proposed by \cite{pitt:shephard:1999} features a second resampling procedure in order to end-up with an equally weighted particle system. This resampling procedure might however severely reduce the accuracy of the filter: \cite{carpenter:clifford:fearnhead:1999} give an example where the accuracy is reduced by a factor of 2; see also \cite{douc:moulines:olsson:2007} for a theoretical proof.}. APF procedures are known to be rather successful when the first-stage distribution is appropriately chosen, which is not always straightforward. The additional computational cost depends mainly on the way the first-stage proposal is designed. The APF method can be mixed with EKF and UKF leading to powerful but computationally involved particle filter; see, \eg, \cite{andrieu:davy:doucet:2003}.

None of the suboptimal methods mentioned above minimise any sensible risk-theoretic criterion and, more annoyingly, both theoretical and practical evidences show that choices which seem to be intuitively correct may lead to performances even worse than that of the plain bootstrap filter (see for example \cite{douc:moulines:olsson:2007} for a striking example). The situation is even more unsatisfactory when the particle filter is driven by a state space dynamic different from that generating the observations, as happens frequently when, \eg, the parameters are not known and need to be estimated or when the model is misspecified.

Instead of trying to guess what a good proposal distribution should be, it seems sensible to follow a more risk-theoretically founded approach. The first step in such a construction consists in choosing a sensible risk criterion, which is not a straightforward task in the SMC context. A natural criterion for SMC would be the variance of the estimate of the posterior mean of a target function (or a set of target functions) of interest, but this approach does not lead to a practical implementation for two reasons. Firstly, in SMC methods, though closed-form expression for the variance at any given current time-step of the posterior mean of any function is available, this variance depends explicitly on all the time steps before the current time. Hence, choosing to minimise the variance at a given time-step would require to optimise all the simulations up to that particular time step, which is of course not practical.
Because of the recursive form of the variance, the minimisation of the conditional variance at each iteration of the algorithm does not necessarily lead to satisfactory performance on the long-run.
Secondly, as for the standard importance sampling algorithm, this criterion is not \emph{function-free}, meaning that a choice of a proposal can be appropriate for a given function, but inappropriate for another.

We will focus in the sequel on function-free risk criteria. A first criterion, advocated in \cite{kong:liu:wong:1994} and \cite{liu:2004} is the \emph{chi-square distance} (CSD) between the proposal and the target distributions, which coincides with the \emph{coefficient of variation} (\CV) of the importance weights. In addition, as heuristically discussed in \cite{kong:liu:wong:1994}, the CSD is related to the \emph{effective sample size}, which estimates the number of i.i.d. samples equivalent to the weighted particle system \footnote{In some situations, the estimated ESS value can be misleading: see the comments of \cite{stephens:donnelly:2000} for a further discussion of this.}.
In practice, the CSD criterion can be estimated, with a complexity that grows linearly with the number of particles, using the empirical \CV\, which can be shown to converge to the CSD as the number of particles tends to infinity. In this paper we show that a similar property still holds in the SMC context, in the sense that the \CV\ still measures a CSD between two distributions $\targjoint$ and $\propjoint$, which are associated with the proposal and target distributions of the particle filter (see Theorem~\ref{th:KL:chi2:convergence}(ii)). Though this result does not come as a surprise, it provides an additional theoretical footing to an approach which is currently used in practice for triggering resampling steps.

Another function-free risk criterion to assess the performance of importance sampling estimators is the KLD between the proposal and the target distributions; see \cite[Chapter 7]{cappe:moulines:ryden:2005}. The KLD shares some of the attractive properties of the CSD; in particular, the KLD may be estimated using the negated empirical \emph{entropy} \entropy\ of the importance weights, whose computational complexity is again linear in the number of particles. In the SMC context, it is shown in Theorem~\ref{th:KL:chi2:convergence}(i) that \entropy\ still converges to the KLD between the same two distributions $\targjoint$ and $\propjoint$ associated with the proposal and the target distributions of the particle filter.

Our methodology to design appropriate proposal distributions is based upon the minimisation of the CSD and KLD between the proposal and the target distributions. Whereas these quantities (and especially the CSD) have been routinely used to detect sample impoverishment and trigger the resampling step \citep{kong:liu:wong:1994}, they have not been used for adapting the simulation parameters in SMC methods. 

We focus here on the auxiliary sampling formulation of the particle filter. In this setting, there are two quantities to optimise: the adjustment multiplier weights (also called \emph{first-stage weights}) and the parameters of the proposal kernel; together these quantites define the mixture used as instrumental distribution in the filter. We first establish a closed-form expression for the limiting value of the CSD and KLD of the auxiliary formulation of the proposal and the target distributions. Using these expressions, we identify a type of auxiliary SMC adjustment multiplier weights which minimise the CSD and the KLD for a given proposal kernel (Proposition~\ref{prop:chi2:optimal:adjfunc}). We then propose several optimisation techniques for adapting the proposal kernels, always driven by the objective of minimising the CSD or the KLD, in coherence with what is done to detect sample impoverishment (see Section \ref{section:adaptive:importance:sampling}). Finally, in the implementation section (Section~\ref{section:implementations:to:state:space:models}), we use the proposed algorithms for  approximating the filtering distributions in several state space models, and show that the proposed optimisation procedure improves the accuracy of the particle estimates and makes them more robust to outlying observations.


\section{Informal presentation of the results}

\subsection{Adaptive importance sampling}
\label{subsec:AdaptiveImportanceSampling}
Before stating and proving rigorously the main results, we discuss informally our findings and introduce the proposed methodology for developing adaptive SMC algorithms. Before entering into the sophistication of sequential methods, we first briefly introduce adaptation of the standard (non-sequential) importance sampling algorithm.

\emph{Importance sampling} (IS) is a general technique to compute expectations of functions \wrt\ a target distribution with density $p(x)$ while only having samples generated from a different distribution---referred to as the \emph{proposal distribution}---with density $q(x)$ (implicitly, the dominating measure is taken to be the Lebesgue measure on $\stsp \define \rset^d$). We sample $\{ \parti{i} \}_{i=1}^{N}$ from the proposal distribution $q$ and compute the unnormalised importance weights $\wgt{i} \define \wgtfuncIS(\parti{i})$,  $i=1, \dots, N$, where $\wgtfuncIS(x) \define p(x)/q(x)$. For any function $f$, the self-normalised importance sampling estimator may be expressed as $\mathrm{IS}_N(f) \define \wgtsum^{-1} \sum_{i=1}^{N} \wgt{i} f(\parti{i})$, where $\wgtsum \define \sum_{j = 1}^{N} \wgt{j}$. As usual in applications of the IS methodology to Bayesian inference, the target density $p$ is known only
up to a normalisation constant; hence we will focus only on a self-normalised version of IS that solely requires the availability of an unnormalised version of $p$ \citep[see][]{geweke:1989}. Throughout the paper, we call a set $\{ \parti{i} \}_{i = 1}^{N}$ of random variables, referred to as \emph{particles} and taking values in $\stsp$, and nonnegative weights $\{ \wgt{i} \}_{i = 1}^{N}$ a \emph{weighted sample} on $\stsp$. Here $N$ is a (possibly random) integer, though we will take it fixed in the sequel. It is well known \citep[see again][]{geweke:1989} that, provided that $f$ is integrable \wrt\ $p$, \ie\ $\int |f(x)| p(x) \, \ud x < \infty$, $\mathrm{IS}_N(f)$ converges, as the number of samples tends to infinity, to the target value $$\E_p[f(X)] \define \int f(x) p(x) \, \ud x \eqsp,$$ for any function $f \in \mathsf{C}$, where $\mathsf{C}$ is the set of functions which are integrable \wrt\ to the target distribution $p$. Under some additional technical conditions, th
 is estimator is also asymptotically normal at rate $\sqrt{N}$; see \cite{geweke:1989}.

It is well known that IS estimators are sensitive to the choice of the proposal distribution. A classical approach consists in trying to minimise the asymptotic variance \wrt\ the proposal distribution $q$. This optimisation is in closed form and leads (when $f$ is a non-negative function) to the optimal choice $q^\ast(x) = f(x) p(x)/ \int f(x) p(x) \, \ud x$, which is, since the normalisation constant is precisely the quantity of interest, rather impractical. Sampling from this distribution can be done by using an accept-reject algorithm, but this does not solve the problem of choosing an appropriate proposal distribution. Note that it is possible to approach this optimal sampling distribution by using the \emph{cross-entropy method}; see \cite{rubinstein:kroese:2004} and \cite{deBoer:kroese:mannor:rubinstein:2005} and the references therein. We will discuss this point later on.

For reasons that will become clear in the sequel, this type of objective is impractical in the sequential context, since the expression of the asymptotic variance in this case is recursive and the optimisation of the variance at a given step is impossible. In addition, in most applications, the proposal density is expected to perform well for a range of typical functions of interest rather than for a specific target function $f$. We are thus looking for \emph{function-free} criteria. The most often used criterion is the CSD between the proposal distribution $q$ and the target distribution $p$, defined as
\begin{align}
\label{eq:definition:chitwo}
\chitwo(p || q) &= \int \frac{\{p(x)-q(x)\}^2}{q(x)} \, \ud x \eqsp, \\
\label{eq:definition:chitwo-1}
                &= \int \wgtfuncIS^2(x) q(x) \, \ud x - 1 \eqsp, \\
\label{eq:definition:chitwo-2}
                &= \int \wgtfuncIS(x) p(x) \, \ud x - 1 \eqsp.
\end{align}
The CSD between $p$ and $q$ may be expressed as the variance of the importance weight function $\wgtfuncIS$ under the proposal distribution, \ie
$$
\chitwo(p || q) = \mathrm{Var}_q[\wgtfuncIS(X)] \eqsp.
$$
This quantity can be estimated by computing the squared coefficient of variation of the unnormalized weights \citep[Section 4]{evans:swartz:1995}:
\begin{equation}
\label{eq:definition:CV}
\CV \big( \left\{ \wgt{i} \right\}_{i=1}^{N} \big) \define N \wgtsum^{-2} \sum_{i=1}^{N} \wgt{i}^2 - 1 \eqsp.
\end{equation}
The $\CV$ was suggested by \cite{kong:liu:wong:1994} as a means for detecting weight degeneracy. If all the weights are equal, then $\CV$ is equal to zero. On the other hand, if all the weights but one are zero, then the coefficient of variation is equal to $N - 1$ which is its maximum value.
From this it follows that using the estimated coefficient of variation for assessing accuracy is equivalent to examining the normalised importance weights to determine if any are relatively large \footnote{Some care should be taken for small sample sizes $N$; the \CV\ can be low because $q$ sample only over a subregion where the integrand is nearly constant, which is not always easy to detect.}.
\cite{kong:liu:wong:1994} showed that the coefficient of variation of the weights $\CV \left( \{ \wgt{i} \right\}_{i=1}^{N} )$ is related to the \emph{effective sample size} (ESS), which is used for measuring the overall efficiency of an IS algorithm:
\begin{equation*}
N^{-1} \mathrm{ESS} \big( \left\{ \wgt{i} \right\}_{i=1}^{N} \big) \define \frac{1}{1 + \CV \big( \left\{ \wgt{i} \right\}_{i=1}^{N} \big)} \to \left\{ 1 + \chitwo(p || q) \right\}^{-1} \eqsp.
\end{equation*}
Heuristically, the ESS measures the number of i.i.d. samples (from $p$) equivalent to the $N$ weighted samples. The smaller the CSD between the proposal and target distributions is, the larger is the ESS. This is why the CSD is of particular interest when measuring efficiency of IS algorithms.

Another possible measure of fit of the proposal distribution is the KLD (also called \emph{relative entropy}) between the proposal and target distributions, defined as
\begin{align}
\label{eq:definition:KL}
\KL( p || q) &\define \int p(x) \log \left( \frac{p(x)}{q(x)} \right) \, \ud x \eqsp, \\  
\label{eq:definition:KL-2}
             &= \int p(x)  \log \wgtfuncIS(x) \, \ud x \eqsp, \\
\label{eq:definition:KL-1}
             &= \int \wgtfuncIS(x) \log \wgtfuncIS(x) \, q(x) \, \ud x \eqsp. 
\end{align}
This criterion can be estimated from the importance weights using the negative \emph{Shannon entropy} \entropy\ of the importance weights:
\begin{equation}
\label{eq:definition:entropy}
\entropy \big( \left\{ \wgt{i} \right\}_{i=1}^{N} \big)  \define \wgtsum^{-1} \sum_{i=1}^{N} \wgt{i} \log \left( N \wgtsum^{-1} \wgt{i} \right) \eqsp.
\end{equation}
The Shannon entropy is maximal when all the weights are equal and minimal when all weights are zero but one. In IS (and especially for the estimation of rare events), the KLD between the proposal and target distributions was thoroughly investigated by \cite{rubinstein:kroese:2004}, and is central in the \emph{cross-entropy} (CE) methodology.

Classically, the proposal is chosen from a family of densities $q_\param$
parameterised by $\param$. Here $\param$ should be thought of as an element of
$\paramsp$, which is a subset of $\rset^k$. The most classical example is the
family of student's $t$-distributions parameterised by mean and covariance.
More sophisticated parameterisations, like mixture of multi-dimensional
Gaussian or Student's $t$-distributions, have been proposed; see, \eg,
\cite{oh:berger:1992}, \cite{oh:berger:1993}, \cite{evans:swartz:1995},
\cite{givens:raftery:1996}, \citet[Chapter 2, Section 2.6]{liu:2004}, and, more
recently, \cite{cappe:douc:guillin:marin:robert:2008} in this issue. In the
sequential context, where computational efficiency is a must, we typically use
rather simple parameterisations, so that the two criteria above can be
(approximatively) solved in a few iterations of a numerical minimisation
procedure.

The optimal parameters for the CSD and the KLD  are those minimising $\param \mapsto \chitwo(p ||q_\param)$ and $\param \mapsto \KL(p || q_\param)$, respectively. In the sequel, we denote by $\param^\ast_{\CSD}$ and $\param^\ast_{\KLD}$ these optimal values.
Of course, these quantities 
cannot be computed in closed form (recall that even the normalisation constant of $p$ is most often unknown; even if it is known, the evaluation of these quantities would involve the evaluation of most often high-dimensional integrals). Nevertheless, it is possible to construct consistent estimators of these optimal parameters. There are two classes of methods, detailed below.

The first uses the fact that the the CSD $\chitwo(p || q_\param)$ and the KLD $\KL(p | q_\param)$ may be approximated by \eqref{eq:definition:CV} and \eqref{eq:definition:entropy}, substituting
in these expressions the importance weights by $\wgt{i}= \wgtfuncIS[\param](\parti[\param]{i})$, $i=1,\dots,N$, where $\wgtfuncIS[\param] \define p / q_\param$ and $\{\parti[\param]{i}\}_{i=1}^{N}$ is a sample from $q_\param$. This optimisation problem formally shares some similarities with the classical minimum chi-square or maximum likelihood estimation, but with
the following important difference: the integrations in \eqref{eq:definition:chitwo} and \eqref{eq:definition:KL} are \wrt\ the proposal distribution $q_\param$ and not the target distribution $p$. As a consequence, the particles $\{ \parti[\param]{i} \}_{i=1}^{N}$ in the definition of the coefficient of variation \eqref{eq:definition:CV} or the entropy \eqref{eq:definition:entropy} of the weights constitute a sample from $q_\param$ and not from the target distribution $p$. As the estimation progresses, the samples used to approach the limiting CSD or KLD can, in contrast to standard estimation procedures, be updated (these samples could be kept fixed, but this is of course inefficient).

The computational complexity of these optimisation problems depends on the way the proposal is parameterised and how the optimisation procedure is implemented. Though the details of the optimisation procedure is in general strongly model dependent, some common principles for solving this optimisation problem can be outlined. Typically, the optimisation is done recursively, \ie\ the algorithm defines a sequence  $\param_{\ell}$, $\ell=0,1,\dots$, of parameters, where $\ell$ is the iteration number. At each iteration, the value of $\param_{\ell}$ is updated by computing a direction $p_{\ell+1}$ in which to step, a step length $\gamma_{\ell+1}$, and setting
\[
\param_{\ell+1} = \param_\ell + \gamma_{\ell+1} p_{\ell+1} \eqsp.
\]
The search direction is typically computed using either Monte Carlo approximation of the finite-difference or (when the quantities of interest are sufficiently regular) the gradient of the criterion. These quantities are used later in conjunction with classical optimisation strategies for computing the step size $\gamma_{\ell+1}$ or normalising the search direction. These implementation issues, detailed in Section~\ref{section:implementations:to:state:space:models}, are model dependent. We denote by $M_{\ell}$ the number of particles used to obtain such an approximation at iteration $\ell$. The number of particles may vary with the iteration index; heuristically there is no need for using a large number of simulations during the initial stage of the optimisation. Even rather crude estimation of the search direction might suffice to drive the parameters towards the region of interest. However, as the iterations go on, the number of simulations should be increased to avoid ``zi
 g-zagging'' when the algorithm approaches convergence. After $L$ iterations, the total number of generated particles is equal to $N = \sum_{\ell=1}^L M_\ell$. Another solution, which is not considered in this paper, would be to use a stochastic approximation procedure, which consists in fixing $M_\ell= M$ and letting the stepsize $\gamma_\ell$ tend to zero. This appealing solution has been successfully used in \cite{arouna:RM:2004}. 
  
The computation of the finite difference or the gradient, being defined as expectations of functions depending on $\param$, can be performed using two different approaches. Starting from definitions \eqref{eq:definition:chitwo-2} and \eqref{eq:definition:KL-2}, and assuming appropriate regularity conditions, the gradient of $\param \mapsto \chitwo(p || q_\param)$ and $\param \mapsto \KL(p || q_\param)$ may be expressed as
\begin{align} 
G_{\CSD}(\param) &\define \nabla_\param \chitwo(p||q_\param) = \int p(x) \nabla_\param \wgtfuncIS[\param](x) \, \ud x 
= \int q_\param(x) \wgtfuncIS[\param](x) \nabla_\param \wgtfuncIS[\param](x) \, \ud x \eqsp,
\label{eq:gradient-CSD} \\
G_{\KLD}(\param) &\define \nabla_\param \KL(p||q_\param) = \int p(x) \nabla_\param \log[ \wgtfuncIS[\param](x) ] \, \ud x  
= \int q_\param(x)   \nabla_\param  \wgtfuncIS[\param](x)  \, \ud x \eqsp. \label{eq:gradient-KLD}
\end{align}
These expressions lead immediately to the following approximations,
\begin{align}
\label{eq:gradient-CSD-empirical}
\hat{G}_{\CSD}(\param) &= M^{-1} \sum_{i=1}^{M} \wgtfuncIS[\param](\parti[\param]{i}) \nabla_\param \wgtfuncIS[\param](\parti[\param]{i}) \eqsp, \\
\label{eq:gradient-KLD-empirical}
\hat{G}_{\KLD}(\param) &= M^{-1} \sum_{i=1}^{M} \nabla_\param  \wgtfuncIS[\param_{\ell}](\parti[\param_\ell]{i}) \eqsp.
\end{align}
There is another way to compute derivatives, which shares some similarities with \emph{pathwise derivative estimates}. Recall that for any $\param \in \paramsp$, one may choose $F_\param$ so that the random variable $F_\param(\noiseletter)$, where  $\noiseletter$ is a vector of independent uniform random variables on $[0,1]^d$, is distributed according to $q_\param$. Therefore, we may express $\param \mapsto \chitwo(p || q_\param)$ and $\param \mapsto \KL(p || q_\param)$ as the following integrals,
\begin{align*}
& \chitwo(p || q_\param)= \int_{[0,1]^d} \compwgtfuncIS[\param](x) \, \ud x \eqsp, \\
& \KL(p || q_\param)= \int_{[0,1]^d} \compwgtfuncIS[\param](x) \log \left[ \compwgtfuncIS[\param](x) \right] \ud x \eqsp,
\end{align*}
where $\compwgtfuncIS[\param](x) \define \wgtfuncIS[\param] \circ F_\param(x)$. Assuming appropriate regularity conditions (\ie\ that $\param \mapsto \wgtfuncIS[\param] \circ F_\param(x)$ is differentiable and that we can interchange the integration and the differentiation), the differential of these quantities \wrt\ $\param$ may be expressed as
\begin{align*}
G_{\CSD}(\param) &= \int_{[0,1]^d} \nabla_\param  \compwgtfuncIS[\param]  (x) \, \ud x \eqsp,\\
G_{\KLD}(\param) &= \int_{[0,1]^d} \left\{ \nabla_\param  \compwgtfuncIS[\param](x) \log [ \compwgtfuncIS[\param] (x) ]  + \nabla_\param \compwgtfuncIS[\param] (x) \right\} \ud x \eqsp.
\end{align*}
For any given $x$, the quantity $\nabla_\param \compwgtfuncIS[\param] (x)$ is the pathwise derivative of the function $\param \mapsto \compwgtfuncIS[\param] (x)$.
As a practical matter, we usually think of each $x$ as a realization of of the output of an ideal random generator. Each $\compwgtfuncIS[\param](x)$ is then the output of the simulation algorithm at parameter $\param$ for the random number $x$. Each $\nabla_\param \compwgtfuncIS[\param](x)$ is the derivative of the simulation output \wrt\ $\param$ with the random numbers held fixed. These two expressions, which of course coincide with \eqref{eq:gradient-CSD} and \eqref{eq:gradient-KLD}, lead to the following estimators,
\begin{align*}
\tilde{G}_{\CSD}(\param) &= M^{-1} \sum_{i=1}^M \nabla_\param  \compwgtfuncIS[\param]  (\noise{i}) \eqsp, \\
\tilde{G}_{\KLD}(\param) &= M^{-1} \sum_{i=1}^M \left\{ \nabla_\param  \compwgtfuncIS[\param](\noise{i}) \log [ \compwgtfuncIS[\param] (\noise{i}) ]  + \nabla_\param \compwgtfuncIS[\param] (\noise{i}) \right\} \eqsp,
\end{align*}
where  
each element of the sequence $\{ \noise{i} \}_{i=1}^{M}$ is a vector on $[0,1]^d$ of independent uniform random variables. It is worthwhile to note that if the number $M_\ell= M$ is kept fixed during the iterations and the uniforms $\{ \noise{i} \}_{i=1}^M$ are drawn once and for all (\ie\ the same uniforms are used at the different iterations), then the iterative algorithm outlined above solves the following problem:
\begin{align}
\label{eq:optimisation-program-CV}
& \param \mapsto \CV\left( \left\{ \compwgtfuncIS[\param] (\noise{i}) \right\}_{i=1}^{M} \right) \eqsp, \\
\label{eq:optimisation-program-entropy}
& \param \mapsto \entropy\left( \left\{ \compwgtfuncIS[\param] (\noise{i}) \right\}_{i=1}^{N} \right) \eqsp.
\end{align}
From a theoretical standpoint, this optimisation problem is very similar to $M$-estimation, and convergence results for $M$-estimators can thus be used under rather standard technical assumptions; see for example \cite{vandervaart:1998}. This is the main advantage of fixing the sample $\{ \noise{i} \}_{i=1}^M$. We use this implementation in the simulations.

Under appropriate conditions, the sequence of estimators $\param^\ast_{\ell,\CSD}$ or $\param^\ast_{\ell,\KLD}$ of these criteria converge, as the number of iterations tends
to infinity, to $\param_{\CSD}^\ast$ or $\param_{\KLD}^\ast$ which minimise the criteria $\param \mapsto \chitwo(p || q_\param)$ and $\param \mapsto \KL( p || q_\param)$, respectively; these theoretical issues are considered in a companion paper.

The second class of approaches considered in this paper is used for minimising the KLD \eqref{eq:optimisation-program-entropy} and is inspired by the cross-entropy method.
This algorithm approximates the minimum $\param_{\mathrm{KLD}}^\ast$ of \eqref{eq:optimisation-program-entropy} by a sequence of pairs of steps, where each step of each pair addresses a simpler optimisation problem. Compared to the previous method, this algorithm is derivative-free and does not require to select a step size. It is in general simpler to implement and avoid most of the common pitfalls of stochastic approximation. Denote by $\param_{0} \in \paramsp$ an initial value. We define recursively the sequence $\{ \param_{\ell} \}_{\ell \geq 0}$ as follows. 
In a first step, we draw a sample $\{\parti[\param_\ell]{i}\}_{i=1}^{M_\ell}$  and evaluate the function
\begin{align}
\label{eq:KL-auxiliaryfunction}
\param \mapsto Q_{\ell}(\param,\param_{\ell}) \define \sum_{i=1}^{M_{\ell}}  \wgtfuncIS[\param_\ell](\parti[\param_\ell]{i})  \log q_{\param} (\parti[\param_\ell]{i}) \eqsp.
\end{align}
In a second step, we choose  $\param_{\ell+1}$ to be the (or any, if there are several) value of $\param \in \paramsp$ that maximises $Q_{\ell}(\param,\param_{\ell})$. As above, the number of particles $M_{\ell}$ is increased during the successive iterations. This procedure ressembles closely the Monte Carlo EM \citep{wei:tanner:1991} for maximum likelihood in incomplete
data models.
The advantage of this approach is that the solution of the maximisation problem $\param_{\ell+1}= \mathrm{argmax}_{\param \in \paramsp} \in  Q_{\ell}(\param,\param_{\ell})$ is often on closed form. In particular, this happens if the distribution $q_\param$ belongs to an \emph{exponential family} (EF) or
is a mixture of distributions of NEF; see \cite{cappe:douc:guillin:marin:robert:2008} for a discussion.
The convergence of this algorithm can be established along the same lines as the convergence of the MCEM algorithm; see \cite{fort:moulines:2003}. As the number of iterations $\ell$ increases, the sequence of estimators $\theta_{\ell}$ may be shown to converge to $\param^\ast_{\KLD}$. These theoretical results are established in a companion paper.

\subsection{Sequential Monte Carlo Methods}
\label{subsec:sequential-Monte-Carlo}
In the sequential context, where the problem consists in simulating from a
\emph{sequence} $\{ p_k \}$ of probability density function, the situation is more
difficult. Let $\stsp_k$ be denote the state space of distribution $p_k$ and note
that this space may vary with $k$, \eg\ in terms of increasing dimensionality. In
many applications, these densities are related to each other by a (possibly random)
mapping, \ie\
$p_{k} = \Psi_{k-1}(p_{k-1})$. In the sequel we focus on the case where there exists
a non-negative function $\ukdens_{k-1}: (\xi,\tilde{\xi}) \mapsto
\ukdens_{k-1}(\xi,\tilde{\xi})$ such that
\begin{equation}
\label{eq:generic-flow}
p_{k}(\tilde{\xi})= \frac{\int \ukdens_{k-1}(\xi,\tilde{\xi}) p_{k-1}(\xi) \, \ud \xi}{\int p_{k-1}(\xi) \int \ukdens_{k-1}(\xi,\tilde{\xi}) \, \ud \tilde{\xi} \, \ud \xi} \eqsp.
\end{equation}
As an example, consider the following generic nonlinear dynamic system described in state space form:
\begin{itemize}
\item \emph{State (system) model}
  \begin{equation}
    \label{eq:state-equation}
    X_k = a( X_{k-1}, U_k) \leftrightarrow  \stackrel{\text{Transition Density}}{\overbrace{\strans(X_{k-1}, X_{k})}} \eqsp,
  \end{equation}
\item \emph{Observation (measurement) model}
  \begin{equation}
    \label{eq:measurement-equation}
    Y_k = b(X_k, V_k) \leftrightarrow \stackrel{\text{Observation Density}}{\overbrace{\olik(X_k, Y_k)}} \eqsp.
  \end{equation}
\end{itemize}
By these equations we mean that each hidden state $X_k$ and data
$Y_k$ are assumed to be generated by nonlinear functions $a(\cdot)$ and $b(\cdot)$,
respectively, of the state and observation noises
$U_k$ and $V_k$. The state and the observation noises $\{ U_k \}_{k \geq 0}$ and $\{
V_k \}_{k \geq 0}$ are assumed to be mutually independent sequences of i.i.d. random
variables. The precise form of the functions and the assumed probability
distributions of the state and observation noises $U_k$ and $V_k$ imply, via a
change of variables, the transition probability density function $\strans(x_{k-1},
x_k)$ and the
observation probability density function $\olik(x_k, y_k)$, the latter being
referred to as the \emph{likelihood of the observation}. With these definitions, the
process $\{ X_k \}_{k \geq 0}$ is Markovian, \ie \ the  conditional probability
density of $X_k$ given the past states $X_{0:k-1} \define (X_0,
\dots, X_{k-1})$ depends exclusively on $X_{k-1}$. This distribution is described by
the density $\strans(x_{k-1}, x_k)$. In addition, the conditional probability
density of $Y_k$ given the states $X_{0:k}$ and the past observations $Y_{0:k-1}$
depends exclusively on $X_k$, and this distribution is captured by the likelihood
$\olik(x_k, y_k)$. We assume further that the initial state $X_0$ is distributed
according to a density function $\prior(x_0)$. Such nonlinear dynamic systems arise
frequently in many areas of science and engineering such as target tracking,
computer vision, terrain referenced navigation, finance, pollution monitoring,
communications, audio engineering, to list only a few.

Statistical inference for the general nonlinear dynamic system above involves
computing the \emph{posterior distribution} of a collection of state variables
$X_{s:s'} \define (X_s, \dots, X_{s'})$
conditioned on a batch $Y_{0:k} = (Y_0, \dots,
Y_k)$ of observations. We denote this posterior distribution by
$\post{s:s'}{k}(X_{s:s'} | Y_{0:k})$. Specific problems include
\emph{filtering}, corresponding to $s = s' = k$,
\emph{fixed lag smoothing}, where $s = s'= k-L$, and \emph{fixed interval smoothing}, with
$s=0$ and $s'=k$. Despite the apparent simplicity of the above problem, the
posterior distributions can be computed in
closed form only in very specific cases, principally, the linear
Gaussian model (where the functions $a(\cdot)$ and $b(\cdot)$ are linear and
the state and observation noises $\{ U_k \}_{k \geq 0}$ and $\{ V_k \}_{k \geq 0}$
are Gaussian) and the discrete \emph{hidden Markov model} (where $X_k$ takes its
values in a finite alphabet). In the vast majority of cases, nonlinearity or
non-Gaussianity render analytic solutions intractable---see
\cite{anderson:moore:1979,kailath:sayed:hassibi:2000,ristic:arulampalam:gordon:2004,cappe:moulines:ryden:2005}.

Starting with the initial, or prior, density function $\prior(x_0)$, and observations $Y_{0:k} = y_{0:k}$, the posterior density $\post{k}{k}(x_{k}|y_{0:k})$ can be obtained using
the following \emph{prediction}-\emph{correction}
recursion~\citep{ho:lee:1964}:
\begin{itemize}
\item \emph{Prediction}
  \begin{equation}
    \label{eq:prediction}
    \post{k}{k-1}(x_{k}|y_{0:k-1}) =
    \post{k-1}{k-1}(x_{k-1}|y_{0:k-1}) \strans(x_{k-1},x_k) \eqsp ,
  \end{equation}
\item \emph{Correction}
  \begin{equation}
    \label{eq:correction}
    \post{k}{k}(x_{k}|y_{0:k}) = \frac{\olik(x_k,y_k) \post{k}{k-1}(x_{k}|y_{0:k-1})}{\lhood{k|k-1}(y_k|y_{0:k-1})} \eqsp ,
  \end{equation}
\end{itemize}
where $\lhood{k|k-1}$ is the predictive distribution of $Y_k$ given the past
observations $Y_{0:k-1}$. For a fixed data realisation, this term is a normalising
constant (independent of the state) and is thus not necessary to compute in standard implementations of SMC methods.

By setting $p_k = \post{k}{k}$, $p_{k-1}= \post{k-1}{k-1}$, and
$$\ukdens_{k-1}(x, x') = \olik(x_k, y_k) \strans(x_{k-1}, x_k) \eqsp,
$$ we conclude that the sequence $\{ \post{k}{k} \}_{k \geq 1}$ of filtering
densities can be generated according to \eqref{eq:generic-flow}.

The case of fixed interval smoothing works entirely analogously: indeed, since 
\begin{equation*}
    \post{0:k}{k-1}(x_{0:k}|y_{0:k-1}) =
    \post{0:k-1}{k-1}(x_{0:k-1}|y_{0:k-1}) \strans(x_{k-1},x_k) 
\end{equation*}
and
\begin{equation*}
    \post{0:k}{k}(x_{k}|y_{0:k}) = \frac{\olik(x_k,y_k) \post{k}{k-1}(x_{0:k}|y_{0:k-1})}{\lhood{k|k-1}(y_k|y_{0:k-1})} \eqsp,
\end{equation*}
the flow $\{ \post{0:k}{k} \}_{k \geq 1}$ of smoothing distributions can be generated according to \eqref{eq:generic-flow} by letting $p_k = \post{0:k}{k}$, $p_{k-1}= \post{0:k-1}{k-1}$, and replacing $\ukdens_{k-1}(x_{0:k-1}, x'_{0:k}) \, \ud x'_{0:k}$ by $\olik(x'_k, y_k) \, \strans(x_{k-1}, x'_k) \, \ud x'_k \\\, \delta_{x_{0:k-1}}(\ud x'_{0:k-1})$,
where $\delta_a$ denotes the Dirac mass located in $a$. Note that this replacement is done formally since the unnormalised transition kernel in question lacks a density in the smoothing mode; this is due to the fact that the Dirac measure is singular \wrt\ the Lebesgue measure. This is however handled by the measure theoretic approach in Section~\ref{section:theoretical:results}, implying that all theoretical results presented in the following will comprise also fixed interval smoothing.

We now adapt the procedures considered in the previous section to the sampling of
densities generated according to \eqref{eq:generic-flow}. Here we focus on a single time-step, and drop from the notation the dependence on $k$ which is irrelevant at this stage. Moreover, set $p_k = \mu$, $p_{k-1} = \nu$, $\ukdens_k = \ukdens$, and assume that we
have at hand a weighted sample $\{ (\parti{N,i}, \wgt{N,i}) \}_{i=1}^{N}$ targeting
$\nu$, \ie, for any $\nu$-integrable function $f$, $\wgtsum^{-1} \sum_{i=1}^{N}
\wgt{N,i} f(\parti{N,i})$ approximates the corresponding integral $\int f(\xi)
\nu(\xi) \, \ud \xi$. A natural strategy for sampling from $\mu$ is to replace $\nu$
in \eqref{eq:generic-flow} by its particle approximation, yielding
\begin{equation*}
\targ(\tilde{\xi})  \define \sum_{i = 1}^{N} \frac{\wgt{N,i} \int \ukdens(\parti{N,i}, \tilde{\xi}) \, \ud \tilde{\xi}}{\sum_{j = 1}^{N} \wgt{N,j} \int \ukdens(\parti{N,j},\tilde{\xi}) \, \ud \tilde{\xi} } \left[ \frac{\ukdens(\parti{N,i}, \tilde{\xi})}{ \int \ukdens(\parti{N,i}, \tilde{\xi}) \, \ud \tilde{\xi}} \right]
\end{equation*}
as an approximation of $\mu$, and simulate $\tilde{M}_N$ new particles from this distribution; however, in many applications direct simulation from $\targ$ is
infeasible without the application of computationally expensive auxiliary accept-reject techniques introduced by \cite{huerzeler:kuensch:1998} and thoroughly analysed by \cite{kuensch:2005}. This difficulty can be overcome by simulating new particles $\{ \partitd{i} \}_{i = 1}^{\tilde{M}_N}$ from the instrumental mixture distribution with density
\[
\instr(\tilde{\xi}) \define \sum_{i=1}^{N} \frac{\wgt{N,i} \adj{i}}{\sum_{j=1}^{N} \wgt{N,j} \adj{j}} \propdens (\parti{N,i}, \tilde{\xi}) \eqsp,
\]
where $\{ \adj{i} \}_{i = 1}^{N}$ are the so-called \emph{adjustment multiplier weights} and $\propdens$ is a Markovian transition density function, \ie, $\propdens(\xi,\tilde{\xi})$ is a nonnegative function and, for any $\xi \in \stsp$, $\int r(\xi,\tilde{\xi}) \, \ud \tilde{\xi}= 1$. If one can guess, based on the new
observation, which particles are most likely to contribute significantly to the posterior, the
resampling stage may be anticipated by increasing (or decreasing) the importance weights. This is the purpose of using the multiplier weights $\adj{i}$. We associate these particles with importance weights $\{ \targ (\partitd{i}) / \instr (\partitd{i}) \}_{i = 1}^{\tilde{M}_N}$. In this setting, a new particle position is simulated from the transition proposal density $\propdens(\parti{N,i}, \cdot)$ with probability proportional to $\wgt{N,i} \adj{i}$. Haplessly, the importance weight $\targ(\partitd{i}) / \instr(\partitd{i})$ is expensive to evaluate since this involves summing over $N$ terms.

We thus introduce, as suggested by \cite{pitt:shephard:1999}, an \emph{auxiliary variable} corresponding to the selected particle, and target instead the probability density
\begin{equation}
\label{eq:definition:auxtarg}
\auxtarg(i,\tilde{\xi}) \define \frac{\wgt{N,i} \int \ukdens(\parti{N,i},\tilde{\xi}) \, \ud \tilde{\xi}}{\sum_{j = 1}^{N} \wgt{N,j} \int \ukdens(\parti{N,j}, \tilde{\xi}) \, \ud \tilde{\xi}} \left[ \frac{\ukdens(\parti{N,i}, \tilde{\xi})}{\int \ukdens(\parti{N,i}, \tilde{\xi}) \, \ud \tilde{\xi}} \right]
\end{equation}
on the product space $\{ 1, \ldots, N \} \times \stsp$. Since $\targ$ is the
marginal distribution of $\auxtarg$ with respect to the particle index $i$,
we may sample from $\targ$ by simulating instead a set $\{ (\ind{i}, \partitd{i}) \}_{i = 1}^{\tilde{M}_N}$ of indices and particle positions from the instrumental distribution
\begin{equation}
\label{eq:definition:auxinstr}
\auxinstr(i,\tilde{\xi}) \define \frac{\wgt{N,i} \adj{i}}{\sum_{j=1}^{N} \wgt{N,j} \adj{j}} \propdens (\parti{N,i},\tilde{\xi})
\end{equation}
and assigning each draw $(\ind{i}, \partitd{i})$ the weight
\begin{equation}
\label{eq:definition:weight}
\wgttd{i} \define \frac{\auxtarg (\ind{i}, \partitd{i})}{\auxinstr (\ind{i}, \partitd{i})} = \adj{\ind{i}}^{-1} \frac{\ukdens(\parti{N,\ind{i}}, \partitd{i})}{\propdens(\parti{N,\ind{i}},\partitd{i})} \eqsp.
\end{equation}
Hereafter, we discard the indices and let $\{ (\partitd{i}, \wgttd{i}) \}_{i = 1}^{\tilde{M}_N}$ approximate the target density $\mu$. Note that setting, for all $i \in \{1, \dots, N\}$, $\adj{i} \equiv 1$ yields the standard bootstrap particle filter presented by \cite{gordon:salmond:smith:1993}. In the sequel, we assume that each adjustment multiplier weight $\adj{i}$ is a function of the particle position $\adj{i}= \adjfunc(\parti{N,i})$, $i \in \{1, \dots, N\}$, and define
\begin{equation}
\label{eq:definition:weightfunction}
\wgtfunc(\xi, \tilde{\xi}) \define \adjfunc^{-1}(\xi) \frac{\ukdens(\xi, \tilde{\xi})}{ \propdens(\xi, \tilde{\xi})}\eqsp,
\end{equation}
so that $\auxtarg(i,\tilde{\xi}) / \auxinstr(i,\tilde{\xi})$ is proportional to $\wgtfunc(\parti{N,i}, \tilde{\xi})$. We will refer to the function $\adjfunc$ as the \emph{adjustment multiplier function}.

\subsection{Risk minimisation for sequential adaptive importance sampling and resampling}
We may expect that the efficiency of the algorithm described above depends highly on the choice of adjustment multiplier weights and proposal kernel.

In the context of state space models, \cite{pitt:shephard:1999} suggested to use an approximation, defined as the value of the likelihood evaluated at the mean of the prior transition, \ie\ $\adj{i} \define \olik \left( \int x' \strans(\parti{N,i},x') \, \ud x', y_k \right)$, where $y_k$ is the current observation, of the predictive likelihood as adjustment multiplier weights. Although this choice of the weight outperforms the conventional bootstrap filter in many applications, as pointed out in \cite{andrieu:davy:doucet:2003}, this approximation
of the predictive likelihood could be very poor and lead to performance even
worse than that of the conventional approach if the dynamic model $\strans(x_{k-1},x_k)$
is quite scattered and the likelihood $\olik(x_k,y_k)$ varies significantly over the prior
$\strans(x_{k-1},x_k)$.

The optimisation of the adjustment multiplier weight was also studied by \cite{douc:moulines:olsson:2007} (see also \cite{olsson:moulines:douc:2007}) who identified adjustment multiplier weights for which the increase of asymptotic variance at a single iteration of the algorithm is minimal. Note however that this optimisation is done using a \emph{function-specific} criterion, whereas we advocate here the use of \emph{function-free} criteria.

In our risk minimisation setting, this means that both the adjustment weights and the proposal kernels need to be adapted. As we will see below, these two problems are in general intertwined; however, in the following it will be clear that the two criteria CSD and KLD behave differently at this point. Because the criteria are rather involved, it is interesting to study their behaviour as the number of particles $N$ grows to infinity. This is done in Theorem \ref{th:KL:chi2:convergence}, which shows that the CSD $\chitwo(\auxtarg||\auxinstr)$ and KLD $\KL(\auxtarg ||\auxinstr)$ converges to
$\chitwo(\targjoint||\propjoint[\adjfunc])$ and $\KL(\targjoint || \propjoint[\adjfunc])$, respectively, where
\begin{equation}
\begin{split} \label{eq:def:targjoint:propjoint}
\targjoint(\xi,\tilde{\xi}) &\define  \frac{\nu(\xi) \, \ukdens(\xi, \tilde{\xi}) }{\iint \nu(\xi) \, \ukdens (\xi, \tilde{\xi}) \, \ud \xi \, \ud \tilde{\xi}} \eqsp, \\
\propjoint[\adjfunc](\xi,\tilde{\xi}) &\define  \frac{\nu(\xi) \adjfunc(\xi) \, \propdens(\xi, \tilde{\xi}) }{\iint \nu(\xi) \adjfunc(\xi) \, \propdens(\xi, \tilde{\xi}) \, \ud \xi \, \ud \tilde{\xi}} \eqsp.
\end{split}
\end{equation}
The expressions \eqref{eq:def:targjoint:propjoint} of the limiting distributions then allow for deriving the adjustment multiplier weight function $\adjfunc$ and the proposal density $\ukdens$ minimising the corresponding discrepancy measures. In absence of constraints (when $\adjfunc$ and $\ukdens$ can be chosen arbitrarily), the optimal solution for both the CSD and the KLD consists in setting $\adjfunc = \adjfunc^\ast$ and $\propdens = \propdens^\ast$, where
\begin{align}
\label{eq:optimal:adjfunc}
&\adjfunc^\ast(\xi) \define  \int \ukdens(\xi,\tilde{\xi}) \, \ud \tilde{\xi} = \int \frac{\ukdens(\xi,\tilde{\xi})}{\propdens(\xi,\tilde{\xi})} \propdens(\xi,\tilde{\xi}) \, \ud \tilde{\xi} \eqsp,\\
\label{eq:optimal:proposalkernel}
&\propdens^\ast(\xi,\tilde{\xi}) \define  \ukdens(\xi,\tilde{\xi}) / \adjfunc^\ast(\xi) \eqsp.
\end{align}
This choice coincides with the so-called \emph{optimal sampling strategy} proposed by \cite{huerzeler:kuensch:1998} and developed further by \cite{kuensch:2005}, which turns out to be \emph{optimal} (in absence of constraints) in our risk-minimisation setting.

\begin{remark}
The limiting distributions $\targjoint$ and $\propjoint[\adjfunc]$ have nice interpretations  within the framework of state space models (see the previous section). In this setting, the limiting distribution $\targjoint$ at time $k$ is the joint distribution $\post{k:k+1}{k+1}$ of the \emph{filtered} couple $X_{k:k+1}$, that is, the distribution of $X_{k:k+1}$ conditionally on the observation record $Y_{0:k+1}$; this can be seen as the asymptotic target distribution of our particle model. Moreover, the limiting distribution $\propjoint$ at time $k$ is only slightly more intricate: Its first marginal corresponds to the filtering distribution at time $k$ reweighted by the adjustment function
$\adjfunc$, which is typically used for incorporating information from the new
observation $Y_{k+1}$. The second marginal of $\propjoint$ is then obtained by
propagating this weighted filtering distribution through the Markovian dynamics of the proposal kernel $\prop$; thus, $\propjoint[\adjfunc]$ describes completely the asymptotic instrumental distribution of the APF, and the two quantities $\KL(\targjoint||\propjoint)$ and
$\chitwo(\targjoint||\propjoint)$ reflect the asymptotic discrepancy between the true model and the particle model at the given time step.
\end{remark}

In presence of constraints on the choice of $\adjfunc$ and $\propdens$, the optimisation of  the adjustment weight function and the proposal kernel density is intertwined. By the so-called \emph{chain rule for entropy} \citep[see][Theorem 2.2.1]{cover:thomas:1991}, we have
\begin{equation*}
\KL(\targjoint || \propjoint[\adjfunc])= \int \frac{\nu(\xi) }{\nu(\adjfunc^\ast)} \adjfunc^\ast(\xi)\log \left( \frac{ \adjfunc^\ast(\xi) / \nu(\adjfunc^\ast)}{\adjfunc(\xi) / \nu(\adjfunc)} \right) \ud \xi 
+ \iint \frac{\nu(\xi)}{\nu(\adjfunc^\ast)}  \ukdens(\xi,\tilde{\xi}) \log \left( \frac{\propdens^\ast(\xi,\tilde{\xi})}{\propdens(\xi,\tilde{\xi})} \right) \ud \xi \, \ud \tilde{\xi} 
\end{equation*}
where $\nu(f) \define \int \nu(\xi) f(\xi) \, \ud \xi $. Hence, if the optimal adjustment function can be chosen freely, then, whatever the choice of the proposal kernel is, the best choice is still $\adjfunc^\ast_{\mathrm{KL},r}  = \adjfunc^\ast$: the best that we can do is to choose $\adjfunc^\ast_{\mathrm{KL},r}$ such that the two marginal distributions $\xi \mapsto \int \targjoint(\xi,\tilde{\xi}) \, \ud \tilde{\xi}$ and $\xi \mapsto \int \propjoint(\xi,\tilde{\xi}) \, \ud \tilde{\xi}$ are identical. If the choices of the weight adjustment function and the proposal kernels are constrained (if, \eg, the weight should be chosen in a pre-specified family of functions or the proposal kernel belongs to a parametric family), nevertheless, the optimisation of $\adjfunc$ and $\propdens$ decouple asymptotically. The optimisation for the CSD does not lead to such a nice decoupling of the adjustment function and the proposal transition; nevertheless, an explicit expression for the adjustment multiplier weights can still be found in
  this case:
\begin{equation}
\label{eq:optimal:adjfuncoptchi}
\adjfunc^\ast_{\chi^2,r}(\xi) \define \sqrt{ \int \frac{ \ukdens^2(\xi, \tilde{\xi})}{\propdens(\xi, \tilde{\xi})}  \, \ud \tilde{\xi}} 
=  \sqrt{ \int \frac{ \ukdens^2(\xi, \tilde{\xi})}{\propdens^2(\xi, \tilde{\xi})} \propdens(\xi,\tilde{\xi}) \, \ud \tilde{\xi}} \eqsp.
\end{equation}
Compared to \eqref{eq:optimal:adjfunc}, the optimal adjustment function for the CSD is the $\Lp{2}$ (rather than the $\Lp{1}$) norm of $\xi \mapsto \ukdens^2(\xi, \tilde{\xi})/\propdens^2(\xi, \tilde{\xi})$. Since $\ukdens(\xi,\tilde{\xi})= \adjfunc^\ast(\xi) \propdens^\ast(\xi,\tilde{\xi})$ (see definitions \eqref{eq:optimal:adjfunc} and \eqref{eq:optimal:proposalkernel}), we obtain, not surprisingly, if we set $\propdens = \propdens^\ast$, $\adjfunc^\ast_{\chi^2,r}(\xi)= \adjfunc^\ast(\xi)$.

Using this risk minimisation formulation, it is possible to select the adjustment weight function as well as the proposal kernel by minimising either the CSD or the KLD criteria. Of course, compared to the sophisticated adaptation strategies considered for adaptive importance sampling, we focus on elementary schemes, the computational burden being quickly a limiting factor in the SMC context.

To simplify the presentation, we consider in the sequel the adaptation of the proposal kernel; as shown above, it is of course possible and worthwhile to jointly optimise the adjustment weight and the proposal kernel, but for clarity we prefer to postpone the presentation of such a technique to a future work. The optimisation of the adjustment weight function is in general rather complex: indeed, as mentioned above, the computation of the optimal adjustment weight function requires the computing of an integral. This integral can be evaluated in closed form only for a rather limited number of models; otherwise, a numerical approximation (based on cubature formulae, Monte Carlo etc) is required, which may therefore incur a quite substantial computational cost. If proper simplifications and approximations are not found (which are, most often, model specific) the gains in efficiency are not necessarily worth the extra cost. In state space (tracking) problems simple and efficient approximations, based either on the EKF or the UKF (see for example \cite{andrieu:davy:doucet:2003} or \cite{shen:dick:brooks:2004}), have been proposed for several models, but the validity of this sort of approximations cannot necessarily be extended to more general models.

In the light of the discussion above, a natural strategy for adaptive design of $\auxinstr$ is to minimise the empirical estimate $\entropy$ (or $\CV$) of the KLD (or CSD) over all proposal kernels belonging to some parametric family $\{ \propdens_\param \}_{\param \in \paramsp}$. This can be done using straightforward adaptations of the two methods described in Section~\ref{subsec:AdaptiveImportanceSampling}.
We postpone a more precise description of the algorithms and implementation issues to after the next section, where more rigorous measure-theoretic notation is introduced and the main theoretical results are stated.
\section{Notations and definitions}

\label{section:notation:and:definitions}
To state precisely the results, we will now use measure-theoretic notation.
In the following we assume that all random variables are defined on a
common probability space $(\Omega, \mathcal{F}, \prob)$ and let, for
any general state space $(\stsp, \alg(\stsp))$, $\meas(\stsp)$ and
$\measfunc(\stsp)$ be the sets of probability measures on
$(\stsp, \alg(\stsp))$ and measurable functions from $\stsp$ to $\R$,
respectively.

A kernel $K$ from $(\stsp,\alg(\stsp))$ to some other
state space $(\stsptd,\alg(\stsptd))$ is called \emph{finite} if
$K(\xi, \stsptd) < \infty$ for all $\xi \in \stsp$ and
\emph{Markovian} if $K(\xi, \stsptd) = 1$ for all $\xi \in \stsp$. Moreover, $K$ induces two operators, one transforming a function
$f \in \measfunc(\stsp \times \stsptd)$ satisfying $\int_{\stsptd}
|f(\xi, \tilde{\xi})| \, K(\xi, \ud \tilde{\xi}) < \infty$ into
another function
\[
\xi \mapsto K(\xi, f) \define \int_{\stsptd} f(\xi, \tilde{\xi}) \, K(\xi, \ud \tilde{\xi})
\]
in $\measfunc(\stsp)$; the other transforms a measure $\nu \in \meas(\stsp)$ into another measure
\begin{equation}
\label{eq:composition-measure-kernel}
A \mapsto \nu K(A) \define \int_{\stsp} K(\xi, A) \, \nu(\ud \xi)
\end{equation}
in $\meas(\stsptd)$. Furthermore, for any probability measure $\mu \in
\meas(\stsp)$ and function $f \in \measfunc(\stsp)$ satisfying
$\int_{\stsp} |f(\xi)| \, \mu(\ud \xi) < \infty$, we write $\mu(f)
\define \int_{\stsp} f(\xi) \, \mu(\ud \xi)$.

The \emph{outer product} of the measure $\gamma$ and the kernel $T$, denoted by $\gamma \otimes T$, is defined as the measure on the product space $\stsp \times \stsptd$, equipped with the product $\sigma$-algebra $\alg(\stsp) \otimes \alg(\stsptd)$, satisfying
\begin{equation}
\label{eq:outer-product-measure-kernel}
\gamma \otimes T(A) \define \iint_{\stsp \times \stsptd} \gamma( \ud \xi) \, T(\xi, \ud \tilde{\xi}) \indic_A(\xi,\xi')
\end{equation}
for any $A \in \alg(\stsp) \otimes \alg(\stsptd)$. For a non-negative function $f \in \measfunc (\stsp)$, we define the modulated measure $\gamma[f]$ on $(\stsp, \alg(\stsp))$ by
\begin{equation}
\label{eq:modulated-measure}
\gamma[f](A) \define \gamma( f \indic_A) \eqsp,
\end{equation}
for any $A \in \alg(\stsp)$.

In the sequel, we will use the following definition. A set $\mathsf{C}$ of real-valued functions on $\stsp$ is said to be \emph{proper} if the following conditions hold: $\textbf{(i)}$ $\mathsf{C}$ is a linear space; $\textbf{(ii)}$ if $g \in \mathsf{C}$ and $f$ is measurable with $|f| \leq |g|$, then $|f| \in \mathsf{C}$; $\textbf{(iii)}$ for all $c \in \R$, the constant function $f \equiv c$ belongs to $\mathsf{C}$.

\begin{definition}
\label{def:consistency}
A weighted sample $\{(\parti{N,i}, \wgt{N,i})\}_{i = 1}^{M_N}$
on $\stsp$ is said to be \emph{consistent} for the probability
measure $\nu \in \meas(\stsp)$ and the set $\mathsf{C}$ if, for any $f \in \mathsf{C}$, as $N \rightarrow \infty$,
\[
\begin{split}
&\wgtsum^{-1} \sum_{i=1}^{M_N} \wgt{N,i} f(\parti{N,i})
\stackrel{\prob}{\longrightarrow} \nu(f) \eqsp,\\ 
&\wgtsum^{-1} \max_{1 \leq i \leq M_N} \wgt{N,i}
\stackrel{\prob}{\longrightarrow} 0 \eqsp,
\end{split}
\]
where $\wgtsum \define \sum_{i=1}^{M_N} \wgt{N,i}$.
\end{definition}
Alternatively, we will sometimes say that the weighted sample in Definition~\ref{def:consistency} \emph{targets} the measure $\nu$.

Thus, suppose that we are given a weighted sample $\{ (\parti{N,i}, \wgt{N,i}) \}_{i = 1}^{M_N}$ targeting $\nu \in \meas(\stsp)$. We wish to transform this sample into a new weighted particle sample approximating the probability measure
\begin{equation} \label{eq:mu:def}
\mu(\cdot) \define \frac{\nu \uk(\cdot)}{\nu \uk(\stsptd)} = \frac{\int_{\stsp} \uk(\xi, \cdot) \, \nu(\ud \xi)}{\int_{\stsp} \uk(\xi', \stsptd) \, \nu(\ud \xi')}
\end{equation}
on some other state space $(\stsptd,\alg(\stsptd))$. Here $\uk$ is a finite transition kernel from $(\stsp,\alg(\stsp))$ to $(\stsptd,\alg(\stsptd))$. As suggested by \cite{pitt:shephard:1999}, an auxiliary variable corresponding to the selected stratum, and target  the measure
\begin{equation}
\label{eq:definition:auxtarg-nodens}
\auxtarg(\{ i \} \times A) \\ \define \frac{\wgt{N,i} \uk(\parti{N,i},
  \stsptd)}{\sum_{j = 1}^{M_N} \wgt{N,j} \uk(\parti{N,j}, \stsptd)} \left[
  \frac{\uk(\parti{N,i}, A)}{\uk(\parti{N,i}, \stsptd)} \right]
\end{equation}
on the product space $\{ 1, \ldots, M_N \} \times \stsp$. Since $\targ$ is the
marginal distribution of $\auxtarg$ with respect to the particle
position, we may sample from $\targ$ by simulating instead a set $\{ (\ind{i}, \partitd{i}) \}_{i = 1}^{\tilde{M}_N}$ of indices and particle positions from the instrumental distribution
\begin{equation}
	\label{eq:definition:auxinst-nodens}
\auxinstr(\{ i \} \times A) \define \frac{\wgt{N,i}
  \adj{i}}{\sum_{j=1}^{M_N} \wgt{N,j} \adj{j}} \prop (\parti{N,i}, A)
\end{equation}
and assigning each draw $(\ind{i}, \partitd{i})$ the weight
\begin{equation*}
\label{eq:definition:weight-nodens}
\wgttd{i} \define \adj{\ind{i}}^{-1} \frac{\ud \uk(\parti{N,\ind{i}}, \cdot)}{\ud \prop(\parti{N,\ind{i}},  \cdot)}(\partitd{i})
\end{equation*}
being proportional to $\ud \auxtarg / \ud \auxinstr (\ind{i}, \partitd{i})$---the formal difference with Equation~\eqref{eq:definition:weight} lies only in the use of
Radon-Nykodym derivatives of the two kernels rather than densities \wrt\
Lebesgue measure. Hereafter, we discard the indices
and take $\{ (\partitd{i}, \wgttd{i}) \}_{i = 1}^{\tilde{M}_N}$ as an approximation of
$\mu$. The algorithm is summarised below.

\begin{algorithm}[ht]
\caption{Nonadaptive APF}
\label{alg:auxiliary:importance:sampling}
\begin{algorithmic}[1]
\Require $\{ (\parti{N,i}, \wgt{N,i}) \}_{i=1}^{M_N}$ targets $\nu$.
\State Draw $\{ \ind{i} \}_{i = 1}^{\tilde{M}_N} \sim
\mathcal{M}( \tilde{M}_N, \{ \wgt{N,j} \adj{j} / \sum_{\ell=1}^{M_N}
\wgt{N,\ell} \adj{\ell} \}_{j = 1}^{M_N})$,
\State simulate $\{ \partitd{i} \}_{i = 1}^{\tilde{M}_N} \sim \bigotimes_{i = 1}^{\tilde{M}_N} \prop(\parti{N,\ind{i}}, \cdot)$,
\State set, for all $i \in \{1,\dots, \tilde{M}_N\}$,
\[
\wgttd{i} \leftarrow
\adj{\ind{i}}^{-1} \ud \uk(\parti{N,\ind{i}}, \cdot) / \ud \prop(\parti{N,\ind{i}},
\cdot)(\partitd{i}) \eqsp.
\]
\State take $\{ (\partitd{i}, \wgttd{i}) \}_{i = 1}^{\tilde{M}_N}$ as
an approximation of $\mu$.
\end{algorithmic}
\end{algorithm}

\section{Theoretical results}
\label{section:theoretical:results}
Consider the following assumptions.
\begin{hyp}{hyp:cons:initial:sample}
The initial sample $\{(\parti{N,i}, \wgt{N,i})\}_{i = 1}^{M_N}$ is
consistent for $(\nu, \mathsf{C})$.
\end{hyp}
\begin{hyp}{hyp:weight:function:assumption}
There exists a function $\adjfunc : \stsp \rightarrow \R^+$ such that
$\adj{i} = \adjfunc(\parti{N,i})$; moreover, $\adjfunc \in \mathsf{C} \cap \Lp{1}(\stsp, \nu)$ and $\uk(\cdot, \stsptd) \in \mathsf{C}$.
\end{hyp}

Under these assumptions we define for $(\xi, \tilde{\xi}) \in \stsp \times \stsptd$ the weight function
\begin{equation}
\label{eq:definition:weightfunction-nodens}
\wgtfunc(\xi, \tilde{\xi}) \define \adjfunc^{-1}(\xi) \frac{\ud \uk(\xi, \cdot)}{\ud \prop(\xi, \cdot)}(\tilde{\xi}) \eqsp,
\end{equation}
so that for every index $i$, $\wgttd{i} = \wgtfunc(\parti{\ind{i}}, \partitd{i})$.
The following result describes how the consistency property is passed through one step of the APF algorithm. A somewhat less general version of this result was also proved in \cite{douc:moulines:olsson:2007} (Theorem~3.1).


\begin{proposition}
Assume {\bf(A\ref{hyp:hyp:cons:initial:sample}, A\ref{hyp:hyp:weight:function:assumption})}. Then the weighted sample $\{(\partitd{i}, \wgttd{i})\}_{i = 1}^{\tilde{M}_N}$ is consistent for $(\nu, \tilde{\mathsf{C}})$, where $\tilde{\mathsf{C}} \define \{ f \in \Lp{1}(\stsptd, \mu), \uk(\cdot, |f|) \in \mathsf{C} \}$.
\end{proposition}
The result above is a direct consequence of Lemma~\ref{lemma:cons:product:space} and the fact that the set $\tilde{\mathsf{C}}$ is proper.

Let $\mu$ and $\nu$ be two probability measures in $\meas(\boldsymbol{\Lambda})$ such that $\mu$ is absolutely continuous with respect to $\nu$. We then recall that the KLD and the CSD are, respectively, given by
\begin{align*}
&\KL(\mu || \nu) \define \int_{\boldsymbol{\Lambda}} \log [\ud \mu / \ud \nu (\lambda)] \, \mu(\ud \lambda) \eqsp, \\ &\chitwo(\mu || \nu) \define \int_{\boldsymbol{\Lambda}} [\ud \mu / \ud \nu (\lambda) - 1]^2 \, \nu(\ud \lambda) \eqsp.
\end{align*}
Define the two probability measures on the product space 
$(\stsp \times \stsptd, \alg(\stsp) \otimes \alg(\stsptd))$:
\begin{equation}
\label{eq:deftargjoint}
\targjoint(A) \define \frac{\nu \otimes \uk}{ \nu \uk (\stsptd)}(A) 
= \frac{\iint_{\stsp \times \stsptd} \nu(\ud \xi) \, \uk(\xi, \ud\xi') \indic_A(\xi,\xi')}{\iint_{\stsp \times \stsptd} \nu(\ud \xi) \, \uk (\xi, \ud \xi')} \eqsp,
\end{equation}
\begin{equation}
\label{eq:defpropjoint}
\propjoint[\adjfunc](A) \define \frac{\nu[\adjfunc] \otimes \prop}{\nu(\adjfunc)}(A) 
= \frac{\iint_{\stsp \times \stsptd} \nu(\ud \xi) \adjfunc(\xi) \, \prop(\xi, \ud \xi') \indic_A(\xi,\xi')}{\iint_{\stsp \times \stsptd} \nu(\ud \xi) \adjfunc(\xi) \, \prop(\xi, \ud \xi')} \eqsp,
\end{equation}
where $A \in \alg(\stsp) \otimes \alg(\stsptd)$ and the outer product $\otimes$ of a measure and a kernel is defined in \eqref{eq:outer-product-measure-kernel}.
\begin{theorem} \label{th:KL:chi2:convergence}
Assume {\bf(A\ref{hyp:hyp:cons:initial:sample}, A\ref{hyp:hyp:weight:function:assumption})}. Then the following holds as $N \rightarrow \infty$.
\begin{enumerate}[(i)]
\item If $\uk(\cdot, |\log \wgtfunc|) \in \mathsf{C} \cap
  \Lp{1}(\stsp, \nu)$, then
\begin{equation}
\label{eq:limiting-KL}
\KL(\auxtarg || \auxinstr) \stackrel{\prob}{\longrightarrow} \KL \left( \left. \targjoint  \right\|  \propjoint[\adjfunc] \right) \eqsp,
\end{equation}

\item If $\uk(\cdot, \wgtfunc) \in \mathsf{C}$, then
\begin{equation}
\label{eq:limiting-chitwo}
\chitwo(\auxtarg || \auxinstr) \stackrel{\prob}{\longrightarrow}  \chitwo \left( \left. \targjoint \right\|  \propjoint[\adjfunc] \right) \eqsp,
\end{equation}
\end{enumerate}
\end{theorem}
Additionally, \entropy\ and \CV, defined in~\eqref{eq:definition:entropy} and \eqref{eq:definition:CV} respectively, converge to the same limits.
\begin{theorem} \label{th:entropy:CV:convergence}
Assume {\bf(A\ref{hyp:hyp:cons:initial:sample}, A\ref{hyp:hyp:weight:function:assumption})}. Then the following holds as $N \rightarrow \infty$.
\begin{enumerate}[(i)]
\item If $\uk(\cdot, |\log \wgtfunc|) \in \mathsf{C} \cap
  \Lp{1}(\stsp, \nu)$, then
\begin{equation}
\label{eq:limiting-entropy}
\entropy( \{ \wgttd{i} \}_{i =1}^{\tilde{M}_N} )\stackrel{\prob}{\longrightarrow} \KL \left( \left. \targjoint  \right\|  \propjoint[\adjfunc] \right) \eqsp.
\end{equation}

\item If $\uk(\cdot, \wgtfunc) \in \mathsf{C}$, then
\begin{equation}
\label{eq:limiting-CV}
 \CV( \{ \wgttd{i} \}_{i = 1}^{\tilde{M}_N}) \stackrel{\prob}{\longrightarrow}  \chitwo \left( \left. \targjoint \right\|  \propjoint[\adjfunc] \right) \eqsp.
\end{equation}
\end{enumerate}
\end{theorem}
Next, it is shown that the adjustment weight function can be chosen so as to minimize the RHS of \eqref{eq:limiting-KL} and \eqref{eq:limiting-chitwo}.
\begin{proposition} \label{prop:chi2:optimal:adjfunc}
Assume {\bf(A\ref{hyp:hyp:cons:initial:sample}, A\ref{hyp:hyp:weight:function:assumption})}. Then the following holds.
\begin{enumerate}[(i)]
\item \label{item:optimal-adjustment-weight-KL}
If $\uk(\cdot, |\log \wgtfunc|) \in \mathsf{C} \cap \Lp{1}(\stsp, \nu)$, then
\begin{equation*}
\operatorname{arg\ min}_{\adjfunc}  \KL \left( \left. \targjoint  \right\|  \propjoint[\adjfunc] \right) \define \adjfuncoptKL \ \text{where}\   \adjfuncoptKL(\xi) \define  \uk(\xi,\stsptd) \eqsp.
\end{equation*}
\item \label{item:optimal-adjustemnt-weight-chi2}
If $\uk(\cdot, \wgtfunc) \in \mathsf{C}$, then
\begin{equation*}
\operatorname{arg\ min}_{\adjfunc} \chitwo \left( \left. \targjoint \right\|  \propjoint[\adjfunc] \right) \define \adjfuncoptchi \ 
\text{where} \ \adjfuncoptchi(\xi) \define \sqrt{ \int_{\stsptd} \frac{\ud \uk(\xi, \cdot)}{\ud \prop(\xi, \cdot)} (\tilde{\xi}) \, \uk (\xi, \ud \tilde{\xi})} \eqsp.
\end{equation*}
\end{enumerate}
\end{proposition}
It is worthwhile to notice that the optimal adjustment weights for the KLD do not depend on the proposal kernel $\prop$. The minimal value $\KL ( \left. \targjoint  \right\|  \propjoint[\adjfuncoptKL] )$ of the limiting KLD is the conditional relative entropy between $\targjoint$ and $\propjoint$.

In both cases, letting $\prop(\cdot, A) = \uk(\cdot, A) / \uk(\cdot, \stsptd)$ yields, as we may expect, the optimal adjustment multiplier weight function $\adjfuncoptKL(\cdot) = \adjfuncoptchi(\cdot) = \uk(\cdot, \stsptd)$, resulting in uniform importance weights $\wgttd{i} \equiv 1$.

It is possible to relate the asymptotic CSD \eqref{eq:limiting-chitwo} between $\auxtarg$ and $\auxinstr$ to the asymptotic variance of the estimator $\wgtsumtd^{-1} \sum_{i=1}^{\tilde{M}_N} \wgttd{i} f(\partitd{i})$ of an expectation $\mu(f)$ for a given integrable target function $f$. More specifically, suppose that $\tilde{M}_N / M_N \rightarrow \ell \in [0, \infty]$ as $N \rightarrow \infty$ and that the initial sample $\{ (\parti{N,i}, \wgt{N,i}) \}_{i = 1}^{M_N}$ satisfies, for all $f$ belonging to a given class $\mathsf{A}$ of functions, the central limit theorem
\begin{equation} \label{eq:ass:asymp:norm}
a_N \wgtsum^{-1} \sum_{i=1}^{M_N} \wgt{N,i} [f(\parti{N,i})-\mu(f)]
\stackrel{\mathcal{D}}{\longrightarrow} \mathcal{N}[0, \sigma^2(f)] \eqsp,
\end{equation}
where the sequence $\{ a_N \}_{N}$ is such that $a_N M_N \rightarrow \beta \in [0, \infty)$ as $N \rightarrow \infty$ and $\sigma : \mathsf{A} \rightarrow \R^+$ is a functional. Then the sample $\{ (\partitd{i}, \wgttd{i}) \}_{i = 1}^{M_N}$ produced in Algorithm~\ref{alg:auxiliary:importance:sampling} is, as showed in \cite[Theorem~3.2]{douc:moulines:olsson:2007}, asymptotically normal for a class of functions $\tilde{\mathsf{A}}$ in the sense that, for all $f \in \tilde{\mathsf{A}}$,
\[
\wgtsumtd^{-1} \sum_{i=1}^{\tilde{M}_N} \wgttd{i} [f(\partitd{i}) - \mu(f)]
\stackrel{\mathcal{D}}{\longrightarrow} \mathcal{N}\{ 0, \tilde{\sigma}^2[\adjfunc](f) \} \eqsp,
\]
where
\begin{equation*} 
\tilde{\sigma}^2[\adjfunc](f) = \sigma^2 \{ \uk[\cdot, f - \mu(f)] \} / [\nu \uk(\cdot, \stsptd)]^2  + \beta \ell^{-1} \nu(\adjfunc \prop \{\cdot, \wgtfunc^2[f - \mu(f)]^2 \} ) \nu(\adjfunc) / [\nu \uk(\stsptd)]^2
\end{equation*}
and, recalling the definition \eqref{eq:modulated-measure} of a modulated measure,
\begin{multline} \label{eq:as:var}
\nu(\adjfunc \prop \{\cdot, \wgtfunc^2[f - \mu(f)]^2 \} ) \nu(\adjfunc) / [\nu \uk(\stsptd)]^2 \\ = \mu^2(|f|) \, \chitwo\{ \targjoint[|f|] / \targjoint(|f|) || \propjoint \} 
\\ - 2 \mu(f) \mu(f_+^{1/2}) \, \chitwo\{ \targjoint[f_+^{1/2}] / \targjoint(f_+^{1/2}) || \propjoint \} \\ + 2 \mu(f) \mu(f_-^{1/2}) \, \chitwo\{ \targjoint[f_-^{1/2}] / \targjoint(f_-^{1/2}) || \propjoint \} \\ 
+ \mu^2(f) \, \chitwo( \targjoint || \propjoint ) + \mu^2(|f|) - \mu^2(f) 
\eqsp.
\end{multline}
Here $f_+ \define \max(f, 0)$ and $f_- \define \max(-f, 0)$ denote the positive and negative parts of $f$, respectively, and $\targjoint(|f|)$ refers to the expectation of the extended function $|f| : (\xi, \tilde{\xi}) \in \stsp \times \stsptd \mapsto |f(\tilde{\xi})| \in \R^+$ under $\targjoint$ (and similarly for $\targjoint(f_{\pm}^{1/2})$). From \eqref{eq:as:var} we deduce that decreasing $\chitwo( \targjoint || \propjoint )$ will imply a decrease of asymptotic variance if the discrepancy between $\targjoint$ and modulated measure $\targjoint[|f|] / \targjoint(|f|)$ is not too large, that is, we deal with a target function $f$ with a regular behavour in the support of $\targjoint(\stsp \times \cdot)$.

\section{Adaptive importance sampling}
\label{section:adaptive:importance:sampling}

\subsection{APF adaptation by minimisation of estimated KLD and CSD over a parametric family}
\label{subsec:adaptive:KLDCSD}

Assume that there exists a random noise variable $\epsilon$, having distribution $\distrn$ on some measurable space $(\boldsymbol{\Lambda}, \alg(\boldsymbol{\Lambda}))$, and a family $\{ F_\param \}_{\param \in \paramsp}$ of mappings from $\stsp \times \boldsymbol{\Lambda}$ to $\stsptd$ such that we are able to simulate $\tilde{\xi} \sim \prop_\param(\xi, \cdot)$, for $\xi \in \stsp$, by simulating $\epsilon \sim \distrn$ and letting $\tilde{\xi} = F_\param(\xi, \epsilon)$. We denote by $\wgtfunc_\param$ the importance weight function associated with $\prop_\param$, see~\eqref{eq:definition:weightfunction-nodens} and set $\wgtfunc_{\param} \circ F_{\param}(\xi, \epsilon) \define \wgtfunc_{\param} (\xi, F_{\param}(\xi, \epsilon))$.

Assume that \refhyp{hyp:cons:initial:sample} holds and suppose that we have simulated, as in the first step of Algorithm~\ref{alg:auxiliary:importance:sampling}, indices $\{ \ind{i} \}_{i = 1}^{\tilde{M}_N}$ and noise variables $\{ \noise{N,i} \}_{i = 1}^{\tilde{M}_N}\sim \distrn^{\otimes \tilde{M}_N}$. Now, keeping these indices and noise variables fixed, we can form an idea of how the KLD varies with $\param$ via the mapping $\param \mapsto \entropy(\{ \wgtfunc_{\param} \circ F_{\param}(\parti{N,\ind{i}}, \noise{N,i})  \}_{i = 1}^{\tilde{M}_N})$. Similarly, the CSD can be studied by using $\CV$ instead of $\entropy$. 
This suggests an algorithm in which
the particles are reproposed using $\prop_{\paramin}$, with
$\paramin = \operatorname{arg\ min}_{\param \in \paramsp}
\entropy(\{ \wgtfunc_{\param} \circ F_{\param}(\parti{N,\ind{i}},
\noise{N,i}) \}_{i = 1}^{\tilde{M}_N})$.

This procedure is summarised in Algorithm~\ref{alg:adaptive:auxiliary:importance:sampling}, and its modification for minimisation of the empirical CSD is straightforward.
\begin{algorithm}[ht]
\caption{Adaptive APF}
\label{alg:adaptive:auxiliary:importance:sampling}
\begin{algorithmic}[1]
\Require \refhyp{hyp:cons:initial:sample}
\State Draw $\{ \ind{i} \}_{i = 1}^{\tilde{M}_N} \sim
\mathcal{M}( \tilde{M}_N, \{ \wgt{N,j} \adj{j} / \sum_{\ell=1}^{M_N}
\wgt{N,\ell} \adj{\ell} \}_{j = 1}^{M_N})$,
\State simulate $\{ \noise{N,i} \}_{i = 1}^{\tilde{M}_N} \sim \distrn^{\otimes \tilde{M}_N}$,
\State
$\paramin \leftarrow \operatorname{arg\ min}_{\param \in \paramsp} \entropy(\{ \wgtfunc_{\param} \circ F_{\param}(\parti{N,\ind{i}}, \noise{N,i}) \}_{i = 1}^{\tilde{M}_N})$,
\State set
\[
\partitd{i} \stackrel{\forall i}{\leftarrow}
F_{\paramin}(\parti{N,\ind{i}}, \noise{N,i})
\]
\State update 
\[
\wgttd{i} \stackrel{\forall i}{\leftarrow}
\wgtfunc_{\paramin}(\parti{N,\ind{i}}, \partitd{i}) \eqsp,
\]
\State let $\{ (\partitd{i}, \wgttd{i}) \}_{i = 1}^{\tilde{M}_N}$ approximate $\mu$.
\end{algorithmic}
\end{algorithm}

\begin{remark}
\label{remark:threshold}
A slight modification of
Algorithm~\ref{alg:adaptive:auxiliary:importance:sampling}, lowering
the added computational burden, is to apply the adaptation mechanism only when the estimated KLD (or CSD) is above a chosen threshold.
\end{remark}

\begin{remark}
It is possible to establish a law of large numbers as well as a central limit theorem for the algorithm above, similarly to what has been done for the nonadaptive auxiliary particle filter in \cite{douc:moulines:olsson:2007} and  \cite{olsson:moulines:douc:2007}.

More specifically, suppose again that \eqref{eq:ass:asymp:norm} holds for similar 
$(\mathsf{A}, \beta, \sigma(\cdot))$ and that $\tilde{M}_N / M_N \rightarrow \ell \in [0, \infty]$ as $N \rightarrow \infty$. Then the sample $\{ (\partitd{i}, \wgttd{i}) \}_{i = 1}^{M_N}$ produced in Algorithm~\ref{alg:adaptive:auxiliary:importance:sampling} is asymptotically normal for a class of functions $\tilde{\mathsf{A}}$ in the sense that, for all $f \in \tilde{\mathsf{A}}$,
\[
\wgtsumtd^{-1} \sum_{i=1}^{\tilde{M}_N} \wgttd{i} [f(\partitd{i}) - \mu(f)]
\stackrel{\mathcal{D}}{\longrightarrow} \mathcal{N}[0, \tilde{\sigma}^2_{\param_\ast}(f)] \eqsp,
\]
where
\begin{equation*}
\tilde{\sigma}^2_{\param_\ast} (f) \define \beta \ell^{-1} \nu(\adjfunc \prop_{\param_\ast} \{\cdot, \wgtfunc^2_{\param_\ast}[f - \mu(f)]^2 \} ) \nu(\adjfunc) / [\nu \uk(\stsptd)]^2  +
\sigma^2( \uk \{\cdot, [f - \mu(f)] \} ) / [\nu \uk(\stsptd)]^2
\end{equation*}
and $\param_\ast$ minimises the asymptotic KLD. The complete proof of this result is however omitted for brevity.
\end{remark}

\subsection{APF adaptation by cross-entropy (CE) methods}
Here we construct an algorithm which selects a proposal kernel from a parametric family in a way that minimises the KLD between the instrumental mixture distribution and the target mixture $\auxtarg$ (defined in~\eqref{eq:definition:auxtarg-nodens}). Thus, recall that we are given an initial sample $\{(\parti{N,i}, \wgt{N,i})\}_{i = 1}^{M_N}$; we then use IS to approximate the target auxiliary distribution $\auxtarg$ by sampling from the instrumental auxiliary distribution
\begin{equation}
\label{eq:definition:auxinst-nodens:param}
\auxinstrparam{\param}(\{ i \} \times A) \define \frac{\wgt{N,i} \adj{i}}{\sum_{j=1}^{M_N} \wgt{N,j} \adj{j}} \prop[\param] (\parti{N,i}, A) \eqsp,
\end{equation}
which is a straightforward modification of~\eqref{eq:definition:auxinst-nodens} where $\prop$
is replaced by $\prop[\param]$, that is, a Markovian kernel from $(\stsp,\alg(\stsp))$ to
$(\stsptd,\alg(\stsptd))$ belonging to the parametric family
$\left\{\prop[\param](\xi, \cdot):  \xi \in  \stsp, \right.$ $\left.\vphantom{\prop[\param]} \param \in \paramsp\right\}$.

We aim at finding the parameter $\param^\ast$ which realises the minimum of $\param \mapsto \KL(\auxtarg||\auxinstrparam{\param})$ over the parameter space $\paramsp$,
where
\begin{equation}
\label{eq:program:minKLD}
\KL(\auxtarg||\auxinstrparam{\param})
= \sum_{i=1}^{M_N} \int_{\stsptd} \log\left( \frac{\ud \auxtarg}{\ud \auxinstrparam{\param}} (i, \tilde{\xi})\right)  \auxtarg(i, \ud \tilde{\xi}) \eqsp.
\end{equation}
Since the expectation in \eqref{eq:program:minKLD} is intractable in most cases,
the key idea is to approximate it iteratively using IS from more and more accurate 
approximations---this idea has been successfully used in CE methods; see e.g. 
\citet{rubinstein:kroese:2004}. At iteration $\ell$, denote by $\param[\ell] \in \paramsp$ the current fit of the parameter. Each iteration of the algorithm is split into two steps:
In the first step we sample, following
Algorithm~\ref{alg:auxiliary:importance:sampling} with $\tilde{M}_N =
\tilde{M}^\ell_N$ and $\prop = \prop[{\param[\ell]}]$, $M_N^\ell$ particles
$\{(\ind[\ell]{i}, \partitd[\ell]{i})\}_{i = 1}^{M_N^\ell}$ from $\auxinstrparam{\param[\ell]}$. Note that the adjustment multiplier weights are kept constant during the iterations, a limitation which is however not necessary. The second step consists in computing the exact solution 
\begin{align}
\label{eq:program:minKLDIS}
\param[\ell+1] \define \argmin_{\param \in \paramsp} \sum_{i=1}^{\tilde{M}^\ell_N}
\frac{\wgttd[\ell]{i}}{\wgtsumtd[\ell]}
\log\left( \frac{\ud \auxtarg}{\ud \auxinstrparam{\param}} (\ind[\ell]{i},\partitd[\ell]{i} )\right) 
\end{align}
to the problem of minimising the Monte Carlo approximation of~\eqref{eq:program:minKLD}. In the case where the kernels $\uk$ and $\prop[\param]$ have densities, denoted by $\ukdens$ and $\propdens[\param]$, respectively, \wrt\ a common reference measure on $(\stsptd, \alg(\stsptd))$, the minimisation program~\eqref{eq:program:minKLDIS} is equivalent to 
\begin{align}
\label{eq:program:minKLDmaxdens}
\param[\ell+1] \define \argmax_{\param \in \paramsp} \sum_{i=1}^{\tilde{M}^\ell_N}
\frac{\wgttd[\ell]{i}}{\wgtsumtd[\ell]}
\log \propdens[\param] (\parti{\ind[\ell]{i}},\partitd[\ell]{i}) \eqsp.
\end{align}
This algorithm is helpful only in situations where the minimisation problem \eqref{eq:program:minKLDIS} is sufficiently simple for allowing for closed-form minimisation;
this happens, for example, if the objective function is a convex combination of concave functions, whose minimum either admits a (simple) closed-form expression or is straightforward to minimise numerically. As mentioned in Section~\ref{subsec:AdaptiveImportanceSampling}, this is generally the case when the function $\propdens[\param](\xi, \cdot)$ belongs to an exponential family for any $\xi \in \stsp$.

Since this optimisation problem closely ressembles the Monte Carlo EM algorithm, all the
implementation details of these algorithms can be readily transposed to that context; see for example \cite{levine:casella:2001}, \cite{eickhof:zhu:amemiya:2004}, and \cite{levine:fan:2004}. Because we use very simple models, convergence occurs, as seen in Section~\ref{section:implementations:to:state:space:models}, within only few iterations. 
When choosing the successive particle sample sizes $\{\tilde{M}^{\ell}_N\}_{\ell=1}^L$, we are facing a trade-off between precision of the approximation \eqref{eq:program:minKLDIS} of~\eqref{eq:program:minKLD} and computational cost. Numerical
evidence typically shows that these sizes may, as high precision is less crucial here than when generating the final population from $\auxinstrparam{\param[L]}$, be relatively small compared to the final size $\tilde{M}_N$. Besides, it is possible (and even theoretically recommended) to increase the number of particles with the iteration index, since, heuristically, high accuracy is less required at the first steps. In the current implementation in Section~\ref{section:implementations:to:state:space:models}, we will show that fixing a priori the total number of iterations and using the same number $\tilde{M}^\ell_N = \tilde{M}_N / L$ of particles at each iteration may provide satisfactory results in a first run.

\begin{algorithm}[ht]
\caption{CE-based adaptive APF}
\label{alg:adaptive:CE}
\begin{algorithmic}[1]
\Require $\{ (\parti{i}, \wgt{i}) \}_{i=1}^{M_N}$ targets $\nu$.
\State Choose an arbitrary $\param[0]$,
\For {$\ell = 0, \ldots, L-1$} \Comment{More intricate criteria are sensible}
\State draw
\[
\{ \ind[\ell]{i} \}_{i = 1}^{\tilde{M}^\ell_N} \sim
\mathcal{M}( \tilde{M}^\ell_N, \{ \wgt{j} \adj{j} / \sum_{n=1}^{M_N}
\wgt{n} \adj{n} \}_{j = 1}^{M_N})\eqsp,
\]
\State simulate $\{ \partitd[\ell]{i} \}_{i = 1}^{\tilde{M}^\ell_N} \sim
\bigotimes_{i = 1}^{\tilde{M}^\ell_N} \prop_{\param[\ell]}(\parti{\ind[\ell]{i}}, \cdot)$,
\State update 
\[
\wgttd{i} \stackrel{\forall i}{\leftarrow}
\wgtfunc_{\param[\ell]}(\parti{\ind[\ell]{i}}, \partitd[\ell]{i}) \eqsp,
\]
\State compute, with available closed-form,
\[
\param[\ell+1] \define \argmin_{\param \in \paramsp} \sum_{i=1}^{\tilde{M}^\ell_N}
\frac{\wgttd[\ell]{i}}{\wgtsumtd[\ell]}
\log\left( \frac{\ud \auxtarg}{\ud \auxinstrparam{\param}} (\ind[\ell]{i},\partitd[\ell]{i} )\right) \eqsp,
\] 
\EndFor
\State run Algorithm~\ref{alg:auxiliary:importance:sampling} with $\prop = \prop_{\param[L]}$.
\end{algorithmic}
\end{algorithm}

\section{Application to state space models}
\label{section:implementations:to:state:space:models}

For an illustration of our findings we return to the framework of state space models in Section~\ref{subsec:sequential-Monte-Carlo} and apply the CE-adaptation-based particle method to \emph{filtering} in nonlinear state space models of type
\begin{equation} \label{eq:nonlinear:state:space:model}
\begin{split}
X_{k + 1} &= m(X_k) + \sigma_w(X_k)W_{k + 1} \eqsp, \quad k \geq 0 \eqsp, \\
Y_k &= X_k + \sigma_v V_k \eqsp, \quad k \geq 0 \eqsp,
\end{split}
\end{equation}
where the parameter $\sigma_v$ and the $\R$-valued functions $(m, \sigma_w)$ are known, and $\{ W_k \}_{k = 1}^\infty$ and $\{ V_k \}_{k = 0}^\infty$ are mutually independent sequences of independent standard normal-distributed variables. In this setting, we wish to approximate the filter distributions $\{ \post{k}{k} \}_{k \geq 0}$, defined in Section~\ref{subsec:sequential-Monte-Carlo} as the posterior distributions of $X_k$ given $Y_{0:k}$ (recall that $Y_{0:k} \define (Y_0, \ldots, Y_k)$), which in general lack closed-form expressions. For models of this type, the optimal weight and density defined in~\eqref{eq:optimal:adjfunc} and~\eqref{eq:optimal:proposalkernel}, respectively, can be expressed in closed-form:
\begin{equation}
\label{eq:optimal:adjfunc:gaussian}
\adjfunc_k^\ast(x)  = \normdens{Y_{k+1}}{m(x)}{\sqrt{\sigma_w^2(x) +  \sigma_v^2}} \eqsp,
\end{equation}
where $\normdens{x}{\mu}{\sigma} \define \exp(- (x-\mu)^2/(2 \sigma^2)) / \sqrt{2 \pi \sigma^2}$  and
\begin{equation}
\label{eq:optimal:proposalkernel:gaussian}
\propdens_k^\ast(x,x') = \normdens{x'}{\meanopt(x,Y_{k+1})}{\stdopt(x)} \eqsp,
\end{equation}
with
\begin{align*}
\meanopt(x,Y_{k+1}) &\define \frac{\sigma_w^2(x) Y_{k + 1}  +  \sigma_v^2 m(x)}{\sigma_w^2(x) +  \sigma_v^2}  \,,\\
\stdopt^2(x) &\define \frac{\sigma_w^2(x)   \sigma_v^2} {\sigma_w^2(x) +  \sigma_v^2}  \,.
\end{align*}
We may also compute the chi-square optimal adjustment multiplier weight function $\Psi^\ast_{\chi^2, \hk}$ when the prior kernel is used as proposal: at time $k$,
\begin{equation} \label{eq:optimal:weights:NL}
\Psi^\ast_{\chi^2,\hk}(x) \propto
\sqrt{\frac{2 \sigma_v^2}{2 \sigma^2_w(x) + \sigma_v^2}} 
\, \exp \left(  - \frac{Y_{k + 1}^2}{\sigma_v^2} + \frac{m(x)}{2 \sigma_w^2(x) + \sigma_v^2}[2 Y_{k + 1} - m(x)] \right)\,.
\end{equation}
We recall from Proposition~\ref{prop:chi2:optimal:adjfunc} that the optimal adjustment weight function for the KLD is given by $\Psi^\ast_{KL,\hk}(x) = \adjfunc_k^\ast(x)$.

In these intentionally chosen simple example we will consider, at each timestep $k$, adaption over the family 
\begin{equation}
	\left\{\prop_\param(x, \cdot)  \define \mathcal{N}(\meanopt(x,Y_{k+1}),
	\param \, \stdopt(x)) : x \in \R, \param > 0 \right\} 
\end{equation}
of proposal kernels. In addition, we keep the adjustment weights constant, that is $\adjfunc(x) = 1$.

The mode of each proposal kernel is centered at the mode of the optimal kernel, and the variance is proportional to the inverse of the Hessian of the optimal kernel at the mode. Let $r_{\param}(x, x') \define \normdens{x'}{\meanopt(x,
Y_{k+1})}{\theta\, \stdopt(x)}$ denote the density of $\prop_\param(x,\cdot)$ \wrt\
the Lebesgue measure. In this setting, at every timestep $k$, a closed-form expression of the KLD between
the target and proposal distributions is available:
\begin{equation}
\KL(\auxtarg || \auxinstrparam{\param}) = \sum_{i=1}^{M_N}  \frac{\wgt{N,i} \adjopt{i}}{\sum_{j=1}^{M_N} \wgt{N,j} \adjopt{j}} 
\left[
\log\left(\frac{\adjopt{i} \wgtsum}{\sum_{j=1}^{M_N} \wgt{N,j} \adjopt{j}}\right)
+ \log \theta  + \frac{1}{2} \left(\frac{1}{\param^2} - 1 \right)
\right] \eqsp,
\label{eq:gaussian:dkl}
\end{equation}
where we set $\adjopt{i} \define \adjfunc^\ast(\parti{N,i})$ and  $\wgtsum = \sum_{i=1}^{M_N} \wgt{N,i}$.

As we are scaling the optimal standard deviation, it is obvious that
\begin{align}
\param^\ast_N &\define \argmin_{\param > 0} \KL(\auxtarg ||
\auxinstrparam{\param}) = 1 \eqsp,
\label{eq:gaussian:optimalparam}
\end{align}
which may also be inferred by straightforward derivation of \eqref{eq:gaussian:dkl}
\wrt\ $\param$. This provides us with a reference to which the parameter values found by our algorithm can be compared. Note that the instrumental distribution $\auxinstrparam{\param^\ast_N}$ differs from the target distribution $\auxtarg$ by the adjustment weights used: recall that every instrumental distribution in the family considered has uniform adjustment
weights, $\adjfunc(x) = 1$, whereas the overall optimal proposal has, since it is equal to the target distribution $\auxtarg$, the optimal weights defined in~\eqref{eq:optimal:adjfunc:gaussian}. This entails that
\begin{equation}
	\KL(\auxtarg || \auxinstrparam{\param^\ast_N})  =
	\sum_{i=1}^{M_N} \wgt{N,i} \frac{\adjopt{i} }{\sum_{j=1}^{M_N} \wgt{N,j}
	\adjopt{j}} \log\left(\frac{\adjopt{i} \wgtsum}{\sum_{j=1}^{M_N} \wgt{N,j}
	\adjopt{j}}\right) \eqsp,
	\label{eq:gaussian:dkl:optfamily}
\end{equation}
which is zero if all the optimal weights are equal.

The implementation of Algorithm~\ref{alg:adaptive:CE} is straightforward as the optimisation program~\eqref{eq:program:minKLDmaxdens} has the following
closed-form solution:
\begin{equation}
	\param[\ell+1]  = \Bigg\{\sum_{i=1}^{M_N^{\ell}}
	\frac{\wgttd{i}^\param[\ell]}{\wgtsumtd^{\param[\ell]} \stdopt^2_{N,\ind{i}^{\param[\ell]}}}
	\left( \partitd[\param_N^{\ell}]{i} -
	\meanopt_{N,\ind{i}^\param[\ell]} \right)^2 \Bigg\}^{1/2} \eqsp,
	\label{eq:gaussian:updateCE}
\end{equation}
where $\meanopt_{N,i} \define \meanopt(\parti{N,i}, Y_{k+1})$ and
$\stdopt^2_{N,i} \define \stdopt^2(\parti{N,i})$.
This is a typical case where the family of proposal kernels allows for efficient
minimisation. Richer families sharing this property may also be used, but here we are voluntarily willing to keep this toy example as simple as possible.

We will study the following special case of the model \eqref{eq:nonlinear:state:space:model}:
\begin{equation*}
m(x) \equiv 0, \quad \sigma_w(x) = \sqrt{\beta_0 + \beta_1 x^2} \eqsp.
\end{equation*}
This is the classical Gaussian \emph{autoregressive conditional heteroscedasticity} (ARCH) \emph{model} observed in noise (see \cite{bollerslev:engle:nelson:1994}). In this case an experiment was conducted where we compared:
\begin{enumerate}
\item[(i)] a plain nonadaptive particle filter
for which $\adjfunc \equiv 1$, that is, the bootstrap particle filter
of \cite{gordon:salmond:smith:1993},
\item[(ii)] an auxiliary filter based on the prior kernel and chi-square optimal weights $\Psi^\ast_{\chi^2, \hk}$,
\item[(iii)] adaptive bootstrap filters with uniform adjustment multiplier weights using numerical minimisation of the empirical CSD and
\item[(iv)] the empirical KLD (Algorithm~\ref{alg:adaptive:auxiliary:importance:sampling}),
\item[(v)] an adaptive bootstrap filter using direct minimisation of $\KL(\auxtarg||\auxinstrparam{\param})$, see~\eqref{eq:gaussian:optimalparam},
\item[(vi)] a CE-based adaptive bootstrap filter, and as a reference,
\item[(vi)] an optimal auxiliary particle filter, i.e. a filter using the optimal weight and proposal kernel defined in \eqref{eq:optimal:adjfunc:gaussian} and~\eqref{eq:optimal:proposalkernel:gaussian}, respectively.
\end{enumerate}

This experiment was conducted for the parameter set $(\beta_0, \beta_1, \sigma_v^2) = (1, 0.99, 10)$, yielding (since $\beta_1 < 1$) a geometrically ergodic ARCH(1) model \citep[see][Theorem~1]{chen:chen:2000}; the noise variance $\sigma^2_v$ is equal to $1/10$ of the stationary variance, which here is equal to $\sigma_s^2 = \beta_0/(1-\beta_1)$, of the state process.

In order to design a challenging test of the
adaptation procedures we set, after having run a hundred burn-in iterations to reach
stationarity of the hidden states, the observations to be constantly equal to $Y_k = 6 \sigma_s$ for every $k \geq 110$. We expect that the bootstrap filter, having a proposal transition
kernel with constant mean $m(x) = 0$, will have a large mean square error (MSE) due a poor number of
particles in regions where the likelihood is significant. We aim at illustrating that the adaptive
algorithms, whose transition kernels have the same mode as the optimal
transition kernel, adjust automatically the variance of the proposals to that of the optimal kernel and reach performances comparable to that of the optimal auxiliary filter.

For these observation records, Figure~\ref{fig:ARCH:MSE} displays MSEs estimates based on $500$ filter means. Each filter used $5,\!000$ particles. The reference values used for the MSE estimates were obtained using the optimal auxiliary particle filter with as many as $500,\!000$ particles. This also provided a set from which the initial particles of every filter were drawn,
hence allowing for initialisation at the filter distribution a few steps before the outlying observations.

The CE-based filter of algorithm~\ref{alg:adaptive:CE} was implemented in its most
simple form, with the inside loop using a constant number
of $M_N^\ell = N/10 = 500$ particles and only $L = 5$ iterations: a simple prefatory
study of the model indicated that the Markov chain $\{\param[\ell]\}_{l \geq
0}$ stabilised around the value reached in the very first step. We set
$\param[0] = 10$ to avoid initialising at the optimal value.

It can be seen in Figure~\ref{fig:ARCH:MSE:fixedseeds} that using the CSD optimal
weights combined with the prior kernel as proposal does not improve on the plain bootstrap filter, precisely because the observations were chosen in such a way that the prior kernel was helpless. On the contrary, Figures~\ref{fig:ARCH:MSE:fixedseeds}
and~\ref{fig:ARCH:MSE:CE} show that the adaptive schemes
perform exactly similarly to the optimal filter: they all success in finding
the optimal scale of the standard deviation, and using uniform adjustment
weights instead of optimal ones does not impact much.

We observe clearly a change of regime, beginning at step $110$, corresponding to the outlying constant
observations. The adaptive filters recover from the
changepoint in one timestep, whereas the bootstrap filter needs several.
More important is that the adaptive filters (as well as the optimal one)
reduce, in the regime of the outlying observations, the MSE of the bootstrap filter by a factor $10$.

Moreover, for a comparison with fixed simulation budget, we ran a bootstrap
filter with $3 N = 15,\!000$ particles
This corresponds to the same simulation budget as the CE-based
adaptive scheme with $N$ particles, which is, in this setting, the fastest of our adaptive algorithms. In our setting, the CE-based filter is measured to expand the plain
bootstrap runtime by a factor $3$, although a basic study of algorithmic
complexity shows that this factor should be closer to $\sum_{\ell=1}^{L} M_N^{\ell} / N = 1.5$---the difference rises from Matlab benefitting from the vectorisation of the plain bootstrap filter, not from the iterative nature of the CE.

The conclusion drawn from Figure~\ref{fig:ARCH:MSE:CE} is that for an equal runtime,
the adaptive filter outperforms, by a factor $3.5$, the bootstrap filter using even three times more particles.

\begin{figure}[htbp]
\centering
\subfigure[Auxiliary filter based on chi-square optimal weights $\Psi^\ast_{\chi^2,\hk}$ and prior kernel $K$ ($\circ$), adaptive filters minimising the empirical KLD ($\ast$) and CSD ($\times$), and reference filters listed below.]{%
\includegraphics[width=.43\textwidth]{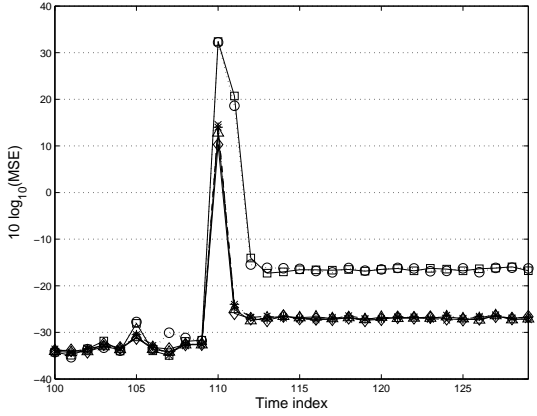}%
\label{fig:ARCH:MSE:fixedseeds}%
}\hfill
\subfigure[CE-based adaption ($\triangle$, dash-dotted line), bootstrap filter with $3 N$ particles ($\square$, dashed line), and reference filters listed below.]{%
\includegraphics[width=.43\textwidth]{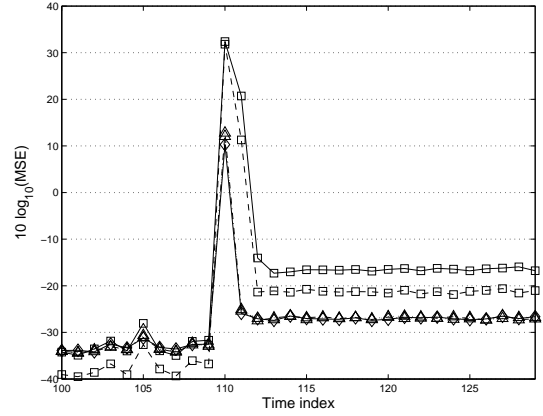}%
\label{fig:ARCH:MSE:CE}%
}
\caption{Plot of MSE performances (on log-scale) on the ARCH model with $(\beta_0, \beta_1, \sigma_v^2) = (1, 0.99, 10)$. Reference filters common to both plots are: the bootstrap filter ($\square$, continuous line), the optimal filter with weights $\Psi^\ast$ and proposal kernel density $\propdens^\ast$ ($\Diamond$), and a bootstrap filter using a proposal with parameter $\param^\ast_N$ minimising the current KLD ($\triangle$, continuous line). The MSE values are computed using $N = 5,\!000$ particles---except for the reference bootstrap using $3 N$ particles ($\square$, dashed line)---and $1,\!000$ runs of each algorithm. }
\label{fig:ARCH:MSE}
\end{figure}

\appendix


\section{Proofs}
\label{section:proofs}

\subsection{Proof of Theorems~\ref{th:KL:chi2:convergence} and~\ref{th:entropy:CV:convergence}}

We preface the proofs of Theorems~\ref{th:KL:chi2:convergence} and~\ref{th:entropy:CV:convergence} with the following two lemmata.

\begin{lemma} \label{lemma:KLD:CSD:identities}
Assume \refhyp{hyp:weight:function:assumption}. Then the following identities hold.
\begin{itemize}
\item[{\it i)}]
	\parbox{1\textwidth}{\[
\KL \left( \left. \targjoint  \right\|  \propjoint[\adjfunc] \right) = \nu \otimes \uk\{\log[ \wgtfunc \nu(\adjfunc) / \nu \uk(\stsptd)] \} / \nu \uk(\stsptd) \eqsp,
\]}
\item[{\it ii)}]\parbox{1\textwidth}{\[
\chitwo \left( \left. \targjoint  \right\| \propjoint[\adjfunc] \right) = \nu(\adjfunc)\, \nu \otimes \uk(\wgtfunc) / [\nu \uk(\stsptd)]^2 - 1 \eqsp.
\]}
\end{itemize}
\end{lemma}

\begin{proof}
We denote by $q(\xi, \xi')$ the Radon-Nikodym derivative of the probability measure $\targjoint$ with respect to $\nu \otimes \prop$ (where the outer product $\otimes$ of a measure and a kernel is defined in \eqref{eq:outer-product-measure-kernel}), that is,
\begin{equation}
\label{eq:density-targjoint}
q(\xi, \xi') \define \frac{ \frac{\ud \uk(\xi,\cdot)}{\ud \prop(\xi,\cdot)}(\xi')}{\iint_{\stsp \times \stsptd} \nu(\ud \xi) \, \uk (\xi, \ud \xi')} \eqsp,
\end{equation}
and by $p(\xi)$ the Radon-Nikodym derivative of the probability measure $\propjoint$ with respect to $\nu \otimes \prop$:
\begin{equation}
\label{eq:density-propjoint}
p(\xi) = \frac{\adjfunc(\xi)}{\nu(\adjfunc)} \eqsp.
\end{equation}
Using the notation above and definition \eqref{eq:definition:weightfunction-nodens} of the weight function $\wgtfunc$, we have
\[
\frac{\wgtfunc(\xi, \xi') \nu(\adjfunc)}{\nu \uk(\stsptd)} = \frac{\nu(\adjfunc) \frac{\ud \uk(\xi, \cdot)}{\ud \prop(\xi,\cdot)} (\xi')}{\adjfunc(\xi) \nu \uk(\stsptd)} = p^{-1}(\xi) q(\xi, \xi') \eqsp.
\]
This implies that
\begin{align*}
\KL \left( \left. \targjoint  \right\|  \propjoint[\adjfunc] \right) &= \iint_{\stsp \times \stsptd} \nu(\ud \xi) \, \prop(\xi, \ud \xi') q(\xi, \xi') \log \left(p^{-1}(\xi) q(\xi, \xi') \right)  \\ &= \nu \otimes \uk\{\log[ \wgtfunc \nu(\adjfunc) / \nu \uk(\stsptd)] \} / \nu \uk(\stsptd) \eqsp,
\end{align*}
which establishes assertion {\it i)}. Similarly, we may write
\begin{align*}
\chitwo \left( \left. \targjoint  \right\| \propjoint[\adjfunc] \right)  
&=  \iint_{\stsp \times \stsptd} \nu(\ud \xi) \, \prop(\xi, \ud \xi') p^{-1}(\xi) q^2(\xi,\xi') - 1 \\
&= \frac{\iint_{\stsp \times \stsptd} \nu(\adjfunc) \, \nu(\ud \xi) \, \prop(\xi, \ud \xi')  \left[ \frac{\ud \uk(\xi,\cdot)}{\ud \prop (\xi,\cdot)} (\xi') \right]^2 \adjfunc^{-1}(\xi)}{[\nu \uk(\stsptd)]^2} - 1 \\
&= \nu(\adjfunc)\, \nu \otimes \uk(\wgtfunc) / [\nu \uk(\stsptd)]^2 - 1 \eqsp,
\end{align*}
showing assertion {\it ii)}.
\end{proof}

\begin{lemma}
\label{lemma:cons:product:space}
Assume {\bf(A\ref{hyp:hyp:cons:initial:sample}, A\ref{hyp:hyp:weight:function:assumption})} and let $\mathsf{C}^\ast \define \{ f \in \measfunc(\stsp \times \stsptd) : \uk(\cdot, |f|) \in \mathsf{C} \cap \Lp{1}(\stsp, \nu) \}$. Then, for all $f \in \mathsf{C}^\ast$, as $N \rightarrow \infty$,
\[
\wgtsumtd^{-1} \sum_{i = 1}^{\tilde{M}_N} \wgttd{i} f(\parti{N,\ind{i}}, \partitd{i}) \stackrel{\prob}{\longrightarrow} \nu  \otimes \uk (f) / \nu \uk(\stsptd) 
\]
\end{lemma}

\begin{proof}
It is enough to prove that
\begin{equation} \label{eq:first:crucial:convergence}
\tilde{M}_N^{-1} \sum_{i = 1}^{\tilde{M}_N} \wgttd{i} f(\parti{N,\ind{i}}, \partitd{i})
\stackrel{\prob}{\longrightarrow} \nu \otimes \uk(f) / \nu(\adjfunc) \eqsp, 
\end{equation}
for all $f \in \mathsf{C}^\ast$; indeed, since the function $f \equiv 1$ belongs to $\mathsf{C}^\ast$ under \refhyp{hyp:weight:function:assumption}, the result of the lemma will follow from \eqref{eq:first:crucial:convergence} by Slutsky's theorem. Define the measure $\varphi(A) \define \nu (\adjfunc \indic_A) / \nu(\adjfunc)$, with $A \in \alg(\stsp)$. By applying Theorem~1 in \cite{douc:moulines:2008} we conclude that the weighted sample $\{ (\parti{N,i}, \adj{i}) \}_{i = 1}^{M_N}$ is consistent for $(\varphi, \{ f\in \Lp{1}(\stsp, \varphi)  : \adjfunc |f| \in \mathsf{C} \} )$. Moreover, by Theorem~2 in the same paper this is also true for the uniformly weighted sample $\{ (\parti{N,\ind{i}}, 1) \}_{i = 1}^{\tilde{M}_N}$ (see the proof of Theorem~3.1 in \cite{douc:moulines:olsson:2007} for details). By definition, for $f \in \mathsf{C}^\ast$, $\varphi \otimes \prop(\wgtfunc |f|) \, \nu(\adjfunc) = \nu \otimes \uk(|f|) < \infty$ and $\adjfunc \prop(\cdot, \wgtfunc |f|) = \uk(\cdot, |f|) \in \mathsf{C}$. Hence, we conclude that $\prop(\cdot, \wgtfunc |f|)$ and thus $\prop(\cdot, \wgtfunc f)$ belong to the proper set $\{ f\in \Lp{1}(\stsp, \varphi)  : \adjfunc |f| \in \mathsf{C} \}$. This implies the convergence
\begin{multline} \label{eq:exp:conv}
\tilde{M}_N^{-1} \sum_{i = 1}^{\tilde{M}_N} \E \left[ \left. \wgttd{i} f(\parti{N,\ind{i}}, \partitd{i}) \right| \mathcal{F}_N \right]  = \tilde{M}_N^{-1}
\sum_{i = 1}^{\tilde{M}_N} \prop(\parti{N,\ind{i}}, \wgtfunc f) \\ \stackrel{\prob}{\longrightarrow} \varphi \otimes \prop(\wgtfunc f) = \nu \otimes \uk(f) / \nu(\adjfunc) \eqsp, \end{multline}
where $\mathcal{F}_N \define \sigma(\{ \parti{N,\ind{i}} \}_{i =
  1}^{\tilde{M}_N})$ denotes the $\sigma$-algebra generated by the
selected particles. It thus suffices to establish that
\begin{equation} \label{eq:crucial:convergence}
\tilde{M}_N^{-1} \sum_{i = 1}^{\tilde{M}_N} \left\{ \E \left[ \left. \wgttd{i} f(\parti{N,\ind{i}}, \partitd{i}) \right| \mathcal{F}_N \right] - \wgttd{i} f(\parti{N,\ind{i}}, \partitd{i}) \right\} \stackrel{\prob}{\longrightarrow}
0 \eqsp,
\end{equation}
and we do this, following the lines of the proof of Theorem~1 in \cite{douc:moulines:2008}, by verifying the two conditions of Theorem~11 in the same work. The sequence $$\left\{ \tilde{M}_N^{-1} \sum_{i = 1}^{\tilde{M}_N} \E \left[ \left. \wgttd{i} |f(\parti{N,\ind{i}}, \partitd{i})| \right| \mathcal{F}_N \right] \right\}_N$$ is tight since it tends to $\nu \otimes \uk(|f|)/\nu(\adjfunc)$ in probability (cf. \eqref{eq:exp:conv}). Thus, the first condition is satisfied. To verify the second condition, take $\epsilon > 0$ and consider, for any $C > 0$, the
decomposition
\begin{multline*} \label{eq:decomp}
\tilde{M}_N^{-1} \sum_{i = 1}^{\tilde{M}_N} \E \left[ \left. \wgttd{i}
    |f(\parti{N,\ind{i}}, \partitd{i})| \indic_{\{ \wgttd{i} |f(\parti{N,\ind{i}}, \partitd{i})| \geq
      \epsilon \}} \right| \mathcal{F}_N \right] \\ \leq
\tilde{M}_N^{-1} \sum_{i = 1}^{\tilde{M}_N} \prop
\left(\parti{N,\ind{i}}, \wgtfunc |f| \indic_{\{ \wgtfunc
    |f| \geq C \}} \right)  + \indic_{ \{ \epsilon
  \tilde{M}_N < C \}} \tilde{M}_N^{-1} \sum_{i = 1}^{\tilde{M}_N} \E
\left[ \left. \wgttd{i} |f(\parti{N,\ind{i}}, \partitd{i})| \right| \mathcal{F}_N
\right] \eqsp.
\end{multline*}
Since $\prop(\cdot, \wgtfunc f)$ belongs to the proper set
$\{ f\in \Lp{1}(\stsp, \varphi) : \adjfunc |f| \in \mathsf{C} \}$,
so does the function $\prop(\cdot, \wgtfunc |f| \indic {\{ \wgtfunc |f| \geq C \}})$. Thus,
since the indicator $\indic \{ \epsilon \tilde{M}_N < C \}$ tends to zero, we conclude that the upper bound above has the limit
$\varphi \otimes \prop(\wgtfunc |f| \indic {\{ \wgtfunc |f| \geq C \}})$; however, by
dominated convergence this limit can be made arbitrarily small by
increasing $C$. Hence
\begin{equation*}
\tilde{M}_N^{-1} \sum_{i = 1}^{\tilde{M}_N} \E \left[ \left. \wgttd{i} |f(\parti{N,\ind{i}}, \partitd{i})| \indic_{\{ \wgttd{i} |f(\parti{N,\ind{i}}, \partitd{i})| \geq \epsilon \}} \right| \mathcal{F}_N \right]
 \stackrel{\prob}{\longrightarrow} 0 \eqsp,
\end{equation*}
which verifies the second condition of Theorem~11 in \cite{douc:moulines:2008}. Thus, \eqref{eq:crucial:convergence} follows.
\end{proof}

\begin{proof}[Proof of Theorem~\ref{th:KL:chi2:convergence}]
We start with {\it i)}. In the light of Lemma~\ref{lemma:KLD:CSD:identities} we establish the limit
\begin{equation} \label{eq:KL:limit}
\KL(\auxtarg || \auxinstr) \stackrel{\prob}{\longrightarrow} \nu \otimes \uk\{ \log[ \wgtfunc \nu(\adjfunc) / \nu \uk(\stsptd)] \} / \nu \uk(\stsptd) 
\eqsp,
\end{equation}
as $N \rightarrow \infty$. Hence, recall the definition (given in Section~\ref{section:theoretical:results}) of the KLD and write, for any index $m \in \{1, \ldots, \tilde{M}_N\}$,
\begin{multline} \label{eq:KL:developed}
\KL(\auxtarg || \auxinstr)  = \sum_{i = 1}^{M_N} \E_{\auxtarg}
\left[ \left. \log \wgtfunc(\parti{N,\ind{m}}, \partitd{m}) \right|
  \ind{m} = i \right] \, \auxtarg(\{ i \} \times \stsptd) \\ + \log
\left[ \frac{\sum_{j=1}^{M_N} \wgt{N,j} \adj{j}}{ \sum_{\ell = 1}^{M_N}
    \wgt{N,\ell} \uk(\parti{N,\ell}, \stsptd)} \right] \eqsp,
\end{multline}
where $\E_{\auxtarg}$ denotes the expectation associated with the random measure $\auxtarg$. For each term of the sum in \eqref{eq:KL:developed} we have
\begin{equation*}
\E_{\auxtarg} \left[ \left. \log \wgtfunc(\parti{N,\ind{m}},
    \partitd{m}) \right| \ind{m} = i \right]  \, \auxtarg(\{ i \}
\times \stsptd)  = \frac{\wgt{N,i} \uk(\parti{N,i}, \log \wgtfunc)}{\sum_{j
    = 1}^{M_N} \wgt{N,i} \uk{}(\parti{N,j}, \stsptd)} \eqsp,
\end{equation*}
and by using the consistency of $\{(\parti{N,i}, \wgt{N,i})\}_{i =
  1}^{M_N}$ (under \refhyp{hyp:cons:initial:sample}) we obtain the limit
\begin{equation*}
\sum_{i = 1}^{M_N} \E_{\auxtarg} \left[ \left. \log
    \wgtfunc(\parti{N,\ind{m}}, \partitd{m}) \right| \ind{m} = i \right]
\, \auxtarg(\{ i \} \times \stsptd) \stackrel{\prob}{\longrightarrow}
\nu \otimes \uk(\log \wgtfunc) / \nu \uk(\stsptd) \eqsp,
\end{equation*}
where we used that $\uk(\cdot, |\log \wgtfunc|) \in \mathsf{C}$ by assumption, implying, since $\mathsf{C}$ is proper, $\uk(\cdot, \log \wgtfunc) \in \mathsf{C}$. Moreover, under \refhyp{hyp:weight:function:assumption}, by the continuous mapping theorem,
\[
\log \left[ \frac{\sum_{j=1}^{M_N} \wgt{N,j} \adj{j}}{ \sum_{\ell =
      1}^{M_N} \wgt{N,\ell} \uk(\parti{N,\ell}, \stsptd)} \right]
\stackrel{\prob}{\longrightarrow} \log[\nu(\adjfunc)/\nu \uk(\stsptd)]
\eqsp,
\]
yielding
\begin{multline*}
\KL(\auxtarg || \auxinstr) \stackrel{\prob}{\longrightarrow} \nu \otimes \uk (\log \wgtfunc) / \nu \uk (\stsptd) + \log[\nu(\adjfunc) / \nu \uk(\stsptd)] \\ = \nu \otimes \uk\{\log[ \wgtfunc \nu(\adjfunc) / \nu \uk(\stsptd)] \} / \nu \uk(\stsptd) \eqsp,
\end{multline*}
which establishes \eqref{eq:KL:limit} and, consequently, \emph{i)}.

To prove \emph{ii)} we show that
\begin{equation} \label{eq:CV:limit}
\chitwo(\auxtarg || \auxinstr) \stackrel{\prob}{\longrightarrow}
\nu(\adjfunc) \, \nu \otimes \uk(\wgtfunc) / [\nu \uk(\stsptd)]^2 - 1 
\end{equation}
and apply Lemma~\ref{lemma:KLD:CSD:identities}. Thus, recall the definition of the CSD and write, for any index $m \in \{1, \ldots, \tilde{M}_N\}$,
\begin{align*}
\chitwo(\auxtarg || \auxinstr) & = \E_{\auxtarg} \left[ \frac{\ud
    \auxtarg}{\ud \auxinstr}(\parti{N,\ind{m}}, \partitd{m}) \right] - 1
\\& = \sum_{i = 1}^{M_N} \E_{\auxtarg} \left[ \left. \frac{\ud
      \auxtarg}{\ud \auxinstr}(\parti{N,\ind{m}}, \partitd{m}) \right|
  \ind{m} = i \right] \, \auxtarg(\{ i \} \times \stsptd)  - 1 \eqsp.
\end{align*}
Here
\begin{multline*}
\E_{\auxtarg} \left[ \left. \frac{\ud \auxtarg}{\ud
      \auxinstr}(\parti{N,\ind{m}}, \partitd{m}) \right| \ind{m} = i
\right]  \, \auxtarg(\{ i \} \times \stsptd) \\ = \wgt{N,i}
\uk(\parti{N,i}, \wgtfunc) \left[ \sum_{j = 1}^{M_N} \wgt{N,i}
  \uk{}(\parti{N,j}, \stsptd) \right]^{-2} \sum_{j = 1}^{M_N} \wgt{N,i}
\adj{i} \eqsp,
\end{multline*}
and using the consistency of $\{(\parti{N,i}, \wgt{N,i})\}_{i = 1}^{M_N}$ yields the limit
\begin{equation*}
\sum_{i = 1}^{M_N} \E_{\auxtarg} \left[ \left. \frac{\ud \auxtarg}{\ud \auxinstr}(\parti{N,\ind{m}}, \partitd{m}) \right| \ind{m} = i \right] \, \auxtarg(\{ i \} \times \stsptd)  \stackrel{\prob}{\longrightarrow} \nu(\adjfunc) \nu \otimes \uk(\wgtfunc) / [\nu \uk(\stsptd)]^2 \eqsp.
\end{equation*}
which proves \eqref{eq:CV:limit}. This completes the proof of \emph{ii)}.
\end{proof}

\begin{proof}[Proof of Theorem~\ref{th:entropy:CV:convergence}]
Applying directly Lemma~\ref{lemma:cons:product:space} for $f = \log \wgtfunc$ (which belongs to $\mathsf{C}^\ast$ by assumption) and the limit \eqref{eq:first:crucial:convergence} for $f \equiv 1$ yields, by the continuous mapping theorem,
\begin{multline*}
\entropy( \{ \wgttd{i} \}_{i = 1}^{\tilde{M}_N}) = \wgtsumtd^{-1} \sum_{i = 1}^{\tilde{M}_N} \wgttd{i} \log \wgttd{i} + \log(\tilde{M}_N \wgtsumtd^{-1}) \\ \stackrel{\prob}{\longrightarrow} \nu \otimes \uk (\log \wgtfunc) / \nu \uk (\stsptd) + \log[\nu(\adjfunc) / \nu \uk(\stsptd)] \\   = \nu \otimes \uk\{\log[ \wgtfunc \nu(\adjfunc) / \nu \uk(\stsptd)] \} / \nu \uk(\stsptd) \eqsp.
\end{multline*}
Now, we complete the proof of assertion \emph{i)} by applying Lemma~\ref{lemma:KLD:CSD:identities}.

We turn to \emph{ii)}. Since $\wgtfunc$ belongs to $\mathsf{C}^\ast$ by assumption, we obtain, by applying Lemma~\ref{lemma:cons:product:space} together with \eqref{eq:first:crucial:convergence},
\begin{multline}
\CV( \{ \wgttd{i} \}_{i = 1}^{\tilde{M}_N} ) = (\tilde{M}_N
\wgtsumtd^{-1}) \wgtsumtd^{-1} \sum_{i = 1}^{\tilde{M}_N} \wgttd{i}^2 - 1
\\ \stackrel{\prob}{\longrightarrow} \nu_{\mathrm{KL}} (\adjfunc) \define \nu(\adjfunc) \, \nu \otimes \uk(\wgtfunc) /
[\nu \uk(\stsptd)]^2 -1 \eqsp.
\end{multline}
From this \emph{ii)} follows via Lemma~\ref{lemma:KLD:CSD:identities}.
\end{proof}

\subsection{Proof of Proposition~\ref{prop:chi2:optimal:adjfunc}}
\label{section:corollary:proof}
Define by $q(\xi) \define \int_{\stsptd} \prop(\xi, \ud \xi') q(\xi,\xi')$ the marginal density of the measure on $(\stsp,\alg(\stsp))$,  $A \in \alg(\stsp) \mapsto \targjoint( A \times \stsptd)$.
We denote by  $q(\xi'|\xi)= q(\xi,\xi')/q(\xi)$ the conditional distribution. By the chain rule of the entropy,
(the entropy of a pair of random variables is the entropy of one plus the conditional entropy of the other),
we may split the KLD between $\targjoint$ and $\propjoint$ as follows,
\begin{equation*}
\KL(\auxtarg || \auxinstr) = \int_{\stsp} \nu( \ud \xi) q(\xi) \log(p^{-1}(\xi) q(\xi)) +
\iint_{\stsp \times \stsptd} \nu( \ud \xi) \prop(\xi, \ud \xi') q(\xi,\xi') \log q(\xi|\xi') \eqsp.
\end{equation*}
The second term in the RHS of the previous equation does not depend on the adjustment multiplier weight $\adjfunc$.
The first term is canceled if we set $p=q$, \ie\ if
\[
\frac{\adjfunc(\xi)}{\nu(\adjfunc)} = \int_\stsptd \prop(\xi, \ud \xi') q(\xi,\xi') = \frac{\uk(\xi,\stsptd)}{\int_\stsp \nu(\ud \xi) \uk(\xi,\stsptd)} \eqsp,
\]
which establishes assertion \emph{i)}.

Consider now assertion \emph{ii)}. 
Note first that
\begin{align}
\label{eq:decomposition-CSD}
&\iint_{\stsp \times \stsptd} \nu(\ud \xi) \, \prop(\xi, \ud \xi') p^{-1}(\xi) q^2(\xi, \xi') - 1  \nonumber\\
&\qquad = \int_\stsp \nu(\ud \xi) \, p^{-1}(\xi) g^2(\xi) -1 \nonumber\\
&\qquad= \nu^2(g) \left\{ \int_\stsp \nu(\ud \xi)  \frac{g^2(\xi)}{p(\xi) \nu^2(g)} - 1 \right\} + \nu^2(g) - 1 \eqsp,
\end{align}
where
\begin{equation*}
\label{eq:definition-g}
g^2(\xi)= \int_\stsptd \prop(\xi, \ud \xi') \, q^2(\xi, \xi') \eqsp.
\end{equation*}
The first term on the RHS of \eqref{eq:decomposition-CSD} is the CSD between the probability distributions associated with the densities $g / \nu(g)$ and $\adjfunc / \nu(\adjfunc)$ with respect to $\nu$. The second term does not depend on $\adjfunc$ and the optimal value of the adjustment multiplier weight is obtained by canceling the first term. This establishes assertion \emph{ii)}.

\section*{Acknowledgements}
The authors are grateful to Prof. Paul Fearnhead for encouragements and useful recommandations, and to the anonymous reviewers for insightful comments and suggestions that improved the presentation of the paper.

\bibliographystyle{apalike}      
\bibliography{AdapSMC_tech_rep_rev.bib}   

\begin{thebibliography}{}

\bibitem[Anderson and Moore, 1979]{anderson:moore:1979}
Anderson, B. D.~O. and Moore, J.~B. (1979).
\newblock {\em Optimal Filtering}.
\newblock Prentice-Hall.

\bibitem[Andrieu et~al., 2003]{andrieu:davy:doucet:2003}
Andrieu, C., Davy, M., and Doucet, A. (2003).
\newblock Efficient particle filtering for jump {M}arkov systems. {A}pplication
  to time-varying autoregressions.
\newblock {\em IEEE Trans. Signal Process.}, 51(7):1762--1770.

\bibitem[Arouna, 2004]{arouna:RM:2004}
Arouna, B. (2004).
\newblock Robbins-monro algorithms and variance reduction in finance.
\newblock {\em Journal of Computational Finance}, 7(2).

\bibitem[Bollerslev et~al., 1994]{bollerslev:engle:nelson:1994}
Bollerslev, T., Engle, R.~F., and Nelson, D. (1994).
\newblock {ARCH} models.
\newblock In Engle, R.~F. and McFadden, D., editors, {\em Handbook of
  Econometrics}, pages 2959--3038. North-Holland.

\bibitem[Capp\'e et~al., 2008]{cappe:douc:guillin:marin:robert:2008}
Capp\'e, O., Douc, R., Guillin, A., Marin, J.~M., and Robert, C.~P. (2008).
\newblock Adaptive importance sampling in general mixture classes.
\newblock {\em Statistics and Computing}, this issue.

\bibitem[Capp\'{e} et~al., 2005]{cappe:moulines:ryden:2005}
Capp\'{e}, O., Moulines, E., and Ryd\'{e}n, T. (2005).
\newblock {\em Inference in Hidden {M}arkov Models}.
\newblock Springer.

\bibitem[Carpenter et~al., 1999]{carpenter:clifford:fearnhead:1999}
Carpenter, J., Clifford, P., and Fearnhead, P. (1999).
\newblock An improved particle filter for non-linear problems.
\newblock {\em IEE Proc. Radar Sonar Navig.}, 146:2--7.

\bibitem[Chen and Chen, 2000]{chen:chen:2000}
Chen, M. and Chen, G. (2000).
\newblock Geometric ergodicity of nonlinear autoregressive models with changing
  conditional variances.
\newblock {\em The Canadian Journal of Statistics}, 28(3):605--613.

\bibitem[Cover and Thomas, 1991]{cover:thomas:1991}
Cover, T.~M. and Thomas, J.~A. (1991).
\newblock {\em Elements of Information Theory}.
\newblock Wiley.

\bibitem[de~Boer et~al., 2005]{deBoer:kroese:mannor:rubinstein:2005}
de~Boer, P.-T., Kroese, D.~P., Mannor, S., and Rubinstein, R.~Y. (2005).
\newblock A tutorial on the cross-entropy method.
\newblock {\em Ann. Oper. Res.}, 134:19--67.

\bibitem[Douc and Moulines, 2008]{douc:moulines:2008}
Douc, R. and Moulines, E. (2008).
\newblock Limit theorems for weighted samples with applications to
  sequential{M}onte {C}arlo.
\newblock {\em Ann. Statist.}, 36.
\newblock To appear.

\bibitem[Douc et~al., 2008]{douc:moulines:olsson:2007}
Douc, R., Moulines, E., and Olsson, J. (2008).
\newblock On the auxiliary particle filter.
\newblock {\em Probab. Math. Statist.}
\newblock To appear.

\bibitem[Doucet et~al., 2001]{doucet:defreitas:gordon:2001}
Doucet, A., {De Freitas}, N., and Gordon, N., editors (2001).
\newblock {\em Sequential {M}onte {C}arlo Methods in Practice}.
\newblock Springer, New York.

\bibitem[Doucet et~al., 2000]{doucet:godsill:andrieu:2000}
Doucet, A., Godsill, S., and Andrieu, C. (2000).
\newblock On sequential {M}onte-{C}arlo sampling methods for {B}ayesian
  filtering.
\newblock {\em Stat. Comput.}, 10:197--208.

\bibitem[Eickhoff et~al., 2004]{eickhof:zhu:amemiya:2004}
Eickhoff, J.~C., Zhu, J., and Amemiya, Y. (2004).
\newblock On the simulation size and the convergence of the {M}onte {C}arlo
  {EM} algorithm via likelihood-based distances.
\newblock {\em Statist. Probab. Lett.}, 67(2):161--171.

\bibitem[Evans and Swartz, 1995]{evans:swartz:1995}
Evans, M. and Swartz, T. (1995).
\newblock Methods for approximating integrals in {S}tatistics with special
  emphasis on {B}ayesian integration problems.
\newblock {\em Statist. Sci.}, 10:254--272.

\bibitem[Fearnhead, 2008]{fearnhead:2008}
Fearnhead, P. (2008).
\newblock Computational methods for complex stochastic systems: a review of
  some alternatives to mcmc.
\newblock {\em Stat. Comput.}, 18:151--171.

\bibitem[Fearnhead and Liu, 2007]{fearnhead:liu:2007}
Fearnhead, P. and Liu, Z. (2007).
\newblock On-line inference for multiple changepoint problems.
\newblock {\em J. Roy. Statist. Soc. Ser. B}, 69(4):590--605.

\bibitem[Fort and Moulines, 2003]{fort:moulines:2003}
Fort, G. and Moulines, E. (2003).
\newblock Convergence of the {M}onte {C}arlo expectation maximization for
  curved exponential families.
\newblock {\em Ann. Statist.}, 31(4):1220--1259.

\bibitem[Fox, 2003]{fox:2003}
Fox, D. (2003).
\newblock Adapting the sample size in particle filters through {KLD}-sampling.
\newblock {\em Int. J. Rob. Res.}, 22(11):985--1004.

\bibitem[Geweke, 1989]{geweke:1989}
Geweke, J. (1989).
\newblock {B}ayesian inference in econometric models using {M}onte-{C}arlo
  integration.
\newblock {\em Econometrica}, 57(6):1317--1339.

\bibitem[Givens and Raftery, 1996]{givens:raftery:1996}
Givens, G. and Raftery, A. (1996).
\newblock Local adaptive importance sampling for multivariate densities with
  strong nonlinear relationships.
\newblock {\em J. Amer. Statist. Assoc.}, 91(433):132--141.

\bibitem[Gordon et~al., 1993]{gordon:salmond:smith:1993}
Gordon, N., Salmond, D., and Smith, A.~F. (1993).
\newblock Novel approach to nonlinear/non-{G}aussian {B}ayesian state
  estimation.
\newblock {\em IEE Proc. F, Radar Signal Process.}, 140:107--113.

\bibitem[Ho and Lee, 1964]{ho:lee:1964}
Ho, Y.~C. and Lee, R. C.~K. (1964).
\newblock A {B}ayesian approach to problems in stochastic estimation and
  control.
\newblock {\em IEEE Trans. Automat. Control}, 9(4):333--339.

\bibitem[Hu et~al., 2008]{hu:schon:ljung:2008}
Hu, X.-L., Schon, T.~B., and Ljung, L. (2008).
\newblock A basic convergence result for particle filtering.
\newblock {\em IEEE Trans. Signal Process.}, 56(4):1337--1348.

\bibitem[H{\"u}rzeler and K{\"u}nsch, 1998]{huerzeler:kuensch:1998}
H{\"u}rzeler, M. and K{\"u}nsch, H.~R. (1998).
\newblock {M}onte {C}arlo approximations for general state-space models.
\newblock {\em J. Comput. Graph. Statist.}, 7:175--193.

\bibitem[Kailath et~al., 2000]{kailath:sayed:hassibi:2000}
Kailath, T., Sayed, A., and Hassibi, B. (2000).
\newblock {\em Linear Estimation}.
\newblock Prentice-Hall.

\bibitem[Kong et~al., 1994]{kong:liu:wong:1994}
Kong, A., Liu, J.~S., and Wong, W. (1994).
\newblock Sequential imputation and {B}ayesian missing data problems.
\newblock {\em J. Am. Statist. Assoc.}, 89(278-288):590--599.

\bibitem[K\"{u}nsch, 2005]{kuensch:2005}
K\"{u}nsch, H.~R. (2005).
\newblock Recursive {M}onte-{C}arlo filters: algorithms and theoretical
  analysis.
\newblock {\em Ann. Statist.}, 33(5):1983--2021.

\bibitem[Legland and Oudjane, 2006]{legland:oudjane:2006}
Legland, F. and Oudjane, N. (2006).
\newblock A sequential algorithm that keeps the particle system alive.
\newblock Technical report, Rapport de recherche 5826, INRIA.

\bibitem[Levine and Casella, 2001]{levine:casella:2001}
Levine, R.~A. and Casella, G. (2001).
\newblock Implementations of the {M}onte {C}arlo {EM} algorithm.
\newblock {\em J. Comput. Graph. Statist.}, 10(3):422--439.

\bibitem[Levine and Fan, 2004]{levine:fan:2004}
Levine, R.~A. and Fan, J. (2004).
\newblock An automated ({M}arkov chain) {M}onte {C}arlo {EM} algorithm.
\newblock {\em J. Stat. Comput. Simul.}, 74(5):349--359.

\bibitem[Liu, 2004]{liu:2004}
Liu, J. (2004).
\newblock {\em {M}onte {Carlo} strategies in scientific computing}.
\newblock Springer.

\bibitem[Oh and Berger, 1992]{oh:berger:1992}
Oh, M.-S. and Berger, J.~O. (1992).
\newblock Adaptive importance sampling in {M}onte {C}arlo integration.
\newblock {\em J. Statist. Comput. Simulation}, 41(3-4):143--168.

\bibitem[Oh and Berger, 1993]{oh:berger:1993}
Oh, M.-S. and Berger, J.~O. (1993).
\newblock Integration of multimodal functions by {M}onte {C}arlo importance
  sampling.
\newblock {\em J. Amer. Statist. Assoc.}, 88(422):450--456.

\bibitem[Olsson et~al., 2007]{olsson:moulines:douc:2007}
Olsson, J., Moulines, E., and Douc, R. (2007).
\newblock Improving the performance of the two-stage sampling particle filter:
  a statistical perspective.
\newblock In {\em Proceedings of the IEEE/SP 14th Workshop on Statistical
  Signal Processing}, pages 284--288, Madison, USA.

\bibitem[Pitt and Shephard, 1999]{pitt:shephard:1999}
Pitt, M.~K. and Shephard, N. (1999).
\newblock Filtering via simulation: Auxiliary particle filters.
\newblock {\em J. Am. Statist. Assoc.}, 94(446):590--599.

\bibitem[Ristic et~al., 2004]{ristic:arulampalam:gordon:2004}
Ristic, B., Arulampalam, M., and Gordon, A. (2004).
\newblock {\em Beyond Kalman Filters: Particle Filters for Target Tracking}.
\newblock Artech House.

\bibitem[Rubinstein and Kroese, 2004]{rubinstein:kroese:2004}
Rubinstein, R.~Y. and Kroese, D.~P. (2004).
\newblock {\em The Cross-Entropy Method}.
\newblock Springer.

\bibitem[Shen et~al., 2004]{shen:dick:brooks:2004}
Shen, C., van~den Hengel, A., Dick, A., and Brooks, M.~J. (2004).
\newblock Enhanced importance sampling: unscented auxiliary particle filtering
  for visual tracking.
\newblock In {\em AI 2004: Advances in artificial intelligence}, volume 3339 of
  {\em Lecture Notes in Comput. Sci.}, pages 180--191. Springer, Berlin.

\bibitem[Shephard and Pitt, 1997]{shephard:pitt:1997}
Shephard, N. and Pitt, M. (1997).
\newblock Likelihood analysis of non-{G}aussian measurement time series.
\newblock {\em Biometrika}, 84(3):653--667.
\newblock Erratum in volume 91, 249--250, 2004.

\bibitem[Soto, 2005]{soto:2005}
Soto, A. (2005).
\newblock Self adaptive particle filter.
\newblock In Kaelbling, L.~P. and Saffiotti, A., editors, {\em Proceedings of
  the 19th International Joint Conferences on Artificial Intelligence (IJCAI)},
  pages 1398--1406, Edinburgh, Scotland.

\bibitem[Stephens and Donnelly, 2000]{stephens:donnelly:2000}
Stephens, M. and Donnelly, P. (2000).
\newblock Inference in molecular population genetics.
\newblock {\em J. R. Stat. Soc. Ser. B Stat. Methodol.}, 62(4):605--655.
\newblock With discussion and a reply by the authors.

\bibitem[Straka and Simandl, 2006]{straka:simandl:2006}
Straka, O. and Simandl, M. (2006).
\newblock Particle filter adaptation based on efficient sample size.
\newblock In {\em Proceedings of the 14th IFAC Symposium on System
  Identification}, pages 991--996, Newcastle, Australia.

\bibitem[Van~der Vaart, 1998]{vandervaart:1998}
Van~der Vaart, A.~W. (1998).
\newblock {\em Asymptotic Statistics}.
\newblock Cambridge University Press.

\bibitem[Wei and Tanner, 1991]{wei:tanner:1991}
Wei, G. C.~G. and Tanner, M.~A. (1991).
\newblock A {M}onte-{C}arlo implementation of the {EM} algorithm and the poor
  man's {D}ata {A}ugmentation algorithms.
\newblock {\em J. Am. Statist. Assoc.}, 85:699--704.

\end{thebibliography}
\end{document}